\documentclass[a4paper,UKenglish,numberwithinsect]{lipics-v2016Modified}

 \usepackage{algorithmicx}
\usepackage{amssymb}
\usepackage{algorithm}
\usepackage{algpseudocode}
 
\usepackage{microtype}

\newtheorem{obs}[theorem]{Observation} 
\newcommand{\trans}[0]{f_{_{y=\beta}}(\mathcal{P})}
\newcommand{\transr}[0]{f_{_{y=\beta}}(\mathcal{P}_r)}
\newcommand{\ntrans}[0]{f_{_{y=*}}(\mathcal{P})}
\newcommand{\transbar}[0]{\overline{f}_{_{y=\beta}}(\mathcal{P}) }
\newcommand{\ntransbar}[0]{\overline{f}_{_{y=*}}(\mathcal{P}) }
\newcommand{\transone}[0]{f(\mathcal{P}_1)}

\newcommand{\transtwo}[0]{f(\mathcal{P}_2)}

\newcommand{\transmeet}[0]{f(\mathcal{P}_1\wedge \mathcal{P}_2)} 
 
\newcommand{\conp}[0]{\cong_{_{\mathcal{P}}}}
\newcommand{\conq}[0]{\cong_{_{\mathcal{Q}}}}
\newcommand{\confp}[0]{\cong_{_{f(\mathcal{P})}}}
\newcommand{\concone}[0]{\cong_{_{f_{_{y=c_1}}(\mathcal{P})}}}
\newcommand{\conctwo}[0]{\cong_{_{f_{_{y=c_2}}(\mathcal{P})}}}
\newcommand{\nconfp}[0]{\cong_{_{f_{_{y=*}}(\mathcal{P})}}}
\newcommand{\conone}[0]{\cong_{_{\mathcal{P}_1}}}
\newcommand{\contwo}[0]{\cong_{_{\mathcal{P}_2}}}
\newcommand{\concap}[0]{\cong_{_{\mathcal{P}_1\wedge \mathcal{P}_2}}}
\newcommand{\gt}[0]{\mathcal{G}(\mathcal{T})}
\newcommand{\gtbar}[0]{\overline{\mathcal{G}(\mathcal{\mathcal{T}})}}
\newcommand{\pred}[0]{\operatorname{Pred}}
\newcommand{\fd}[0]{f_{_D}}

\title{A fix-point characterization of Herbrand equivalence of expressions in data flow frameworks}
\titlerunning{A fix-point characterization of Herbrand equivalence} 

\author[1]{Jasine Babu}
\author[2]{K. Murali Krishnan}
\author[2]{Vineeth Paleri}
\affil[1]{Indian Institute of Technology Palakkad, India\\
  \texttt{jasine@iitpkd.ac.in}}
\affil[2]{National Institute of Technology Calicut, India\\
  \texttt{\{kmurali,vpaleri\}@nitc.ac.in}}
\authorrunning{J.\,Babu, K.\,Murali Krishnan and V.\,Paleri} 

\Copyright{Jasine Babu and K. Murali Krishnan and Vineeth Paleri}

\subjclass{F.3.2 -- please refer to \url{http://www.acm.org/about/class/ccs98-html}}

\EventEditors{}
\EventNoEds{0}
\EventLongTitle{Arxiv Preprint}
\EventShortTitle{ArXiv Preprint}
\EventAcronym{ArXiv}
\EventYear{2017}
\EventDate{}
\EventLocation{}
\EventLogo{}
\SeriesVolume{}
\ArticleNo{Page}

\begin{document}
\maketitle
\begin{abstract}
The problem of determining Herbrand equivalence of terms at each program point in a data flow framework is a central and well studied question in program
analysis. Most of the well-known algorithms for the computation of Herbrand 
equivalence in data flow frameworks \cite{GULWANI2007, RRPai2016, Ruthing1999} proceed via iterative fix-point computation on some abstract lattice of short expressions \textit{relevant}
to the given flow graph.  However the mathematical definition of Herbrand equivalence is based on a meet over all path characterization over 
the (infinite) set of all possible expressions (see ~\cite[p.~393]{Steffen90}).  The aim of this paper is to develop a lattice theoretic fix-point formulation of Herbrand equivalence on 
the (infinite) concrete lattice defined over the set of all terms constructible from variables, constants and operators of a program.  
The present characterization uses an axiomatic formulation of the notion of Herbrand congruence and defines the (infinite) concrete lattice of Herbrand congruences.  
Transfer functions and non-deterministic assignments are formulated as monotone functions over this concrete lattice.  
Herbrand equivalence is defined as the maximum fix point of a composite transfer function defined over an appropriate product lattice of the above 
concrete lattice.  A  re-formulation of the classical meet-over-all-paths 
definition of Herbrand equivalence (~\cite[p.~393]{Steffen90}) in the above lattice theoretic framework is also presented and is proven to 
be equivalent to the new lattice theoretic fix-point characterization. 
\end{abstract}
\section{Introduction}
A data flow framework is an abstract representation of a program, used in program analysis and compiler optimizations \cite{aho2007compilers}.  
As detection of semantic equivalence of expressions at each point in a program is unsolvable \cite{Kam1977}, all known algorithms try to detect a weaker, 
syntactic notion of expression equivalence  over the set of all possible expressions, called  \emph{Herbrand equivalence}.   
Stated informally, Herbrand equivalence treats operators as uninterpreted functions, 
and two expressions are considered equivalent if they are obtained by applying the same operator on equivalent operands \cite{GULWANI2007,Muller2005,Ruthing1999,Steffen90}.   

The pioneering work of Kildall \cite{Kildall1973}, which essentially is an abstract interpretation \cite{Cousot1977} of terms, 
showed that at each program point, Herbrand equivalence of expressions that occur in a program could be computed 
using an iterative refinement algorithm. The algorithm models each iteration as the application of a monotone function over a meet semi-lattice, and terminates at a 
fix-point of the function \cite{Kam1976,Kam1977}.  Subsequently, several problems in program analysis 
have been shown to be solvable using iterative fix-point computation on lattice frameworks.  (see \cite{aho2007compilers} for examples).  Several algorithms for 
computation of Herbrand equivalence of program expressions also were proposed in the literature \cite{Alpern1988,GULWANI2007,RRPai2016,Rosen1988,Ruthing1999,Steffen90}.  

Although algorithmic computation to detect Herbrand equivalence among \textit{expressions that appear in a program} proceeds via iterative fix-point computation on an 
abstract lattice framework,  the classical mathematical definition of Herbrand equivalence uses a meet over all path formulation over 
the (infinite) set of all possible expressions (see~\cite[p.~393]{Steffen90}).  The main difficulty in constructing a fix-point based definition for Herbrand equivalence of expressions
at each program point is that it requires consideration of 
all program paths and equivalence of all expressions - including expressions not appearing in the program. 
Consequently, such a characterization of Herbrand equivalence cannot be achieved without resorting to set theoretic machinery.  

It may be noted that, while the algorithm presented by Steffen et.~al.~\cite{Steffen90} uses an iterative fix-point
computation method, their definition of Herbrand equivalences was essentially a meet over all paths (MOP) formulation. 
Though the MOP based definition of Herbrand equivalences given by Steffen et.~al.~\cite{Steffen90} is known since 1990,
proving the completeness of iterative fix-point based algorithms using this definition is non-trivial.  
For instance, the algorithm proposed by the same authors~\cite{Ruthing1999} was proven to be 
incomplete \cite{GULWANI2007} after several years, though it was initially accepted to be complete.
In comparison with an MOP based definition of Herbrand equivalences, a fix-point characterization will 
render completeness proofs of iterative fix-point algorithms for computing the equivalence of program expressions simpler.   
The completeness proofs would now essentially involve establishing an equivalence preserving continuous homomorphism from the 
infinite concrete lattice of all Herbrand congruences to the finite abstract lattice 
of congruences of expressions that are \textit{relevant} to the program, and proceed via induction.  

In this paper, we develop a lattice theoretic fix-point characterization of Herbrand equivalences at each program point in a data flow framework.    
We define the notion of a congruence relation on the set of all expressions, and show that the set of all congruences form a complete lattice.  Given a data flow
framework with $n$ program points, we show how to define a continuous composite transfer function over the $n$ fold product of the above lattice, such that  
the maximum fix-point of the function yields the set of Herbrand equivalence classes at various program points.  This characterization is then
shown to be equivalent to a meet over all paths formulation of expression equivalence over the same lattice framework.  

Section~\ref{sec:terminology} introduces the basic notation.  Sections \ref{sec:cong-1} to \ref{sec:non-det-asg} develop the basic theory of congruences and
transfer functions, including non-deterministic assignment functions.  Section~\ref{sec:data-flow-frameworks} and  Section~\ref{sec:Herbrand} deal with the application of the formalism of congruences
to derive a fix-point characterization of Herbrand equivalence at each program point.  Section~\ref{sec:MOP} describes a meet-over-all-paths formulation for expression equivalence and 
establishes the equivalence between the fix-point characterization and the meet-over-all-paths formulation.    
\section{Terminology}\label{sec:terminology}
Let $C$ be a countable set of constants and $X$ be a countable set of variables.  For simplicity, we assume that the set of operators $Op=\{+\}$.  (More operators can be added without any difficulty).   
The set of all terms over $C\cup X$, $\mathcal{T}=\mathcal{T}(X,C)$ is defined by $t::=c \mid x \mid (t+t)$, with $c\in C$ and $x\in X$.   (Parenthesis is 
avoided when there is no confusion.) Let $\mathcal{P}$ be a partition of $\mathcal{T}$.  Let $[t]_{\mathcal{P}}$ (or simply $[t]$ when there is no confusion) denote the 
equivalence class containing the term $t\in \mathcal{T}$.  If $t'\in [t]_{\mathcal{P}}$,  we write $t\cong_{\mathcal{P}}t'$ (or simply $t\cong t'$). 
Note that $\cong$ is reflexive, symmetric and transitive.   For any $A \subseteq  \mathcal{T}$, $A(x)$ denotes the set of all terms in $A$ in which 
the variable $x$ appears and $\overline{A}(x)$ denotes the set of all terms in $A$ in which $x$ does not appear.  
In particular, for any $x\in X$,  $\mathcal{T}(x)$ is the set of all terms containing the variable $x$ and $\overline{\mathcal{T}}(x)$ denotes
the set of terms in which $x$ does not appear.
\begin{definition}[Substitution]\label{def:substitution}
For $t,\alpha\in \mathcal{T}$, $x\in X$, substitution of $x$ with $\alpha$ in $t$, denoted by $t[x\leftarrow \alpha]$ is defined as follows:  
\begin{enumerate}
 \item If $t=x$, then $t[x\leftarrow \alpha]=\alpha$. 
 \item If $t\notin \mathcal{T}(x)$, $t[x\leftarrow \alpha]=t$
 \item If $t=t_1+t_2$ then $t[x\leftarrow \alpha]=t_1[x\leftarrow \alpha]+t_2[x\leftarrow \alpha]$.  
\end{enumerate}
\end{definition}
\begin{remark}
In the rest of the paper, complete proofs for statements that are left unproven in the main text are provided in the appendix. Proofs for some elementary properties of 
lattices, which are used in the main text, are also given in the appendix, to make the presentation self contained.
\end{remark}
\section{Congruences of Terms}\label{sec:cong-1}
We define the notion of congruence (of terms).  The notion of congruence will be useful later to model equivalence of terms at various program points in a data flow framework.
\begin{definition}[Congruence of Terms]\label{def:cong}
Let $\mathcal{P}$ be a partition of $\mathcal{T}$.  $\mathcal{P}$ is a Congruence (of terms) if the following conditions hold:
\begin{enumerate}
 \item For each $c,c'\in C$, if $c\neq c'$ then  $c\ncong c'$. (No two distinct constants are congruent).  
 \item For $t,t',s,s'\in \mathcal{T}$, $t'\cong t$ and $s'\cong s$ if and only if  $t'+s'\cong t+s$.  (Congruences respect operators).    
 \item For any $c\in C$, $t\in \mathcal{T}$, if $t\cong c$ then either $t=c$ or $t\in X$.  (The only non-constant terms that are allowed to be congruent to a constant are  variables).  
\end{enumerate}
\end{definition}
The motivation for the definition of congruence is the following.  Given the representation of a program in a data flow framework (to be defined later), 
we will associate a congruence to each program point at each iteration.  
Each iteration refines the present congruence at each program point.  We will see later that this process of refinement leads to a well defined ``fix point congruence'' 
at each program point.  We will see that this fix point congruence captures Herbrand equivalence at that program point.    
\begin{definition}
The set of all congruences over $\mathcal{T}$ is denoted by $\mathcal{G}(\mathcal{T})$.  We first note the substitution property of congruences.    
\end{definition}
\begin{obs}[Substitution Property]\label{obs:substitution}
 Let $\mathcal{P}$ be a congruence over $\mathcal{T}$.  Then,  for each $\alpha, \beta \in \mathcal{T}$, $\alpha\cong \beta$ if and only if for all 
 $x\in X$ and $t\in \mathcal{T}$, $t[x\leftarrow \alpha]\cong t[x\leftarrow \beta]$.
\end{obs}
\begin{proof}
One direction is easy.  If for all $x\in X$ and $t\in \mathcal{T}$, $t[x\leftarrow \alpha]\cong t[x\leftarrow \beta]$, then setting $t=x$ we get 
$\alpha=t[x\leftarrow \alpha]\cong t[x\leftarrow \beta]=\beta$.  
Conversely, suppose $\alpha\cong \beta$.  Let $x\in X$ and $t\in \mathcal{T}$ be chosen arbitrarily. To prove $t[x\leftarrow \alpha]\cong t[x\leftarrow \beta]$, we use induction.  
If $t\notin \mathcal{T}(x)$, then $t[x\leftarrow \alpha]=t\cong t=t[x\leftarrow \beta]$.  If $t=x$, then $t[x\leftarrow \alpha]=\alpha\cong \beta =t[x\leftarrow \beta]$.
Otherwise, if $t=t_1+t_2$, then $t[x\leftarrow \alpha]=t_1[x\leftarrow \alpha]+t_2[x\leftarrow \alpha]
\cong t_1[x\leftarrow \beta]+t_2[x\leftarrow \beta] \text{ (by induction hypothesis)}=t[x\leftarrow \beta]$.
\end{proof}
The following observation is a direct consequence of condition (3) of the definition of congruence.  
\begin{obs}\label{obs:const2}
 Let $\mathcal{P}\in \gt$, $c\in C$ and let $t\in \mathcal{T}(y)$, with $t\neq y$.  Then $c\ncong_{_{\mathcal{P}}} t[y\leftarrow c]$.
\end{obs}
We define a binary {\em confluence operation} on the set of congruences, $\mathcal{G}(\mathcal{T})$.
A confluence operation transforms a pair of congruences into a congruence.  
\begin{definition}[Confluence]\label{def:confluence}
Let $\mathcal{P}_1=\{A_i\}_{i\in I}$ and $\mathcal{P}_2=\{B_j\}_{j\in J}$ 
be two congruences.  For all $i\in I$ and $j\in J$, define  $C_{i,j}=A_i\cap B_j$.   The confluence  of $\mathcal{P}_1$ and $\mathcal{P}_2$ is defined by:
\\ $\mathcal{P}_1\wedge \mathcal{P}_2=\{C_{i,j}: i\in I,j\in J, C_{i,j}\neq \emptyset\}$.  
\end{definition}
\begin{theorem}\label{thm:confluence}
If $\mathcal{P}_1$ and $\mathcal{P}_2$ are congruences, then $\mathcal{P}_1\wedge \mathcal{P}_2$ is a congruence.  
\end{theorem}
\section{Structure of Congruences}\label{sec:cong-2}
In this section we will define an ordering on the set $\gt$ and then extend it it to a complete lattice.     
\begin{definition}[Refinement of a Congruence]
Let $\mathcal{P}$, $\mathcal{P'}$ be congruences over $\mathcal{T}$.  We say $P\preceq P'$ (read $P$ is a refinement of $P'$) if for each equivalence class $A\in \mathcal{P}$, there exists
an equivalence class $A'\in \mathcal{P'}$ such that $A\subseteq A'$.  
\end{definition}
\begin{definition}\label{def:bot}
The partition in which each term in $\mathcal{T}$ belongs to a different class is defined as: $\bot=\{ \{t\} : t\in \mathcal{T} \}$.    
\end{definition}
The following observation is a direct consequence of the definition of $\bot$.  
\begin{obs}\label{obs:bot}
 $\bot$ is a congruence.  Moreover for any $\mathcal{P}\in \gt$, $\bot \wedge \mathcal{P}=\bot$.  
\end{obs}
\begin{definition}\label{def:semilattice}
A partially ordered set $(P,\leq)$ is a meet semi-lattice if every pair of elements $a,b\in P$ has a greatest lower bound (called the meet of $a$ and $b$).   
\end{definition}
\begin{lemma}\label{lem:poset}
$(\mathcal{G}(\mathcal{T}),\preceq)$ is a meet semi-lattice with meet operation $\wedge$ and bottom element $\bot$. 
\end{lemma}
The following lemma extends the meet operation to arbitrary non-empty collections of congruences.  The proof relies on the axiom of choice.  
\begin{lemma}\label{lem:subsetInfimum}
 Every non-empty subset of $(\mathcal{G}(\mathcal{T}),\preceq)$ has a greatest lower bound.  
\end{lemma} 
Next, we extend the meet semi-lattice $(\gt,\preceq)$ by artificially adding a top element $\top$, so that the greatest lower bound of the empty set is also well defined.  
\begin{definition}\label{def:PBar}
The lattice $(\gtbar,\preceq, \bot, \top)$  is defined over the set  $\gtbar=\gt \cup \{\top \}$ 
with $\mathcal{P}\preceq \top$ for each $\mathcal{P}\in \gtbar$.  
In particular, $\top$ is the greatest lower bound of $\emptyset$ and 
$\top \wedge \mathcal{P}=\mathcal{P}$ for every $\mathcal{P}\in \gtbar$.  
\end{definition}
Hereafter, we will be referring the element $\top$ as a congruence.  It follows from Lemma~\ref{lem:subsetInfimum} and the above definition that every subset 
of $\gtbar$  has a greatest lower bound.  Since a meet semi-lattice in which every subset has a greatest lower bound is a complete lattice (Theorem~\ref{Athm:complete}), we have:
\begin{theorem}\label{thm:complete}
$(\gtbar, \preceq, \bot, \top)$ is a complete lattice.   
\end{theorem}
\begin{definition}[Infimum]\label{def:sup-inf}
Let $S=\{\mathcal{P}_i\}_{i\in I}$ be an arbitrary collection of congruences in $\gtbar$ ($S$ may be empty or may contain $\top$).  
The infimum of the set $S$, denoted by $\bigwedge S$ or $\bigwedge_{i\in I}\mathcal{P}_i$, 
is defined as the greatest lower bound of the set $\{\mathcal{P}_i\}_{i\in I}$.  
\end{definition}
\section{Transfer Functions}\label{sec:transfun-1}
We now define a class of unary operators on $\gt$ called {\em transfer functions}.  A transfer function specifies how the assignment of a term to a variable transforms the congruence  
before the assignment into a new one.  
\begin{definition}[Transfer Functions]\label{def:trans-function}
Let $y\in X$ and $\beta\in \overline{\mathcal{T}}(y)$. (Note that $y$ does not appear in $\beta$).  Let $\mathcal{P}=\{A_i\}_{i\in I} $ be an arbitrary congruence. The transfer function 
$\trans: \mathcal{G}(\mathcal{T})\longrightarrow \mathcal{G}(\mathcal{T})$ transforms $\mathcal{P}$ to a congruence $\trans$ given by the following:  
\begin{itemize}
 \item For each $i\in I$, let $B_{i}=\{t \in \mathcal{T}$ : $t[y\leftarrow \beta]\in A_{i}\}$.
 \item $\trans=\{B_{i}$ : $i\in I$, $B_i\neq \emptyset\}$.
\end{itemize}
\end{definition}
  \begin{figure}[h] 
  \begin{center}
 \includegraphics[scale=0.75]{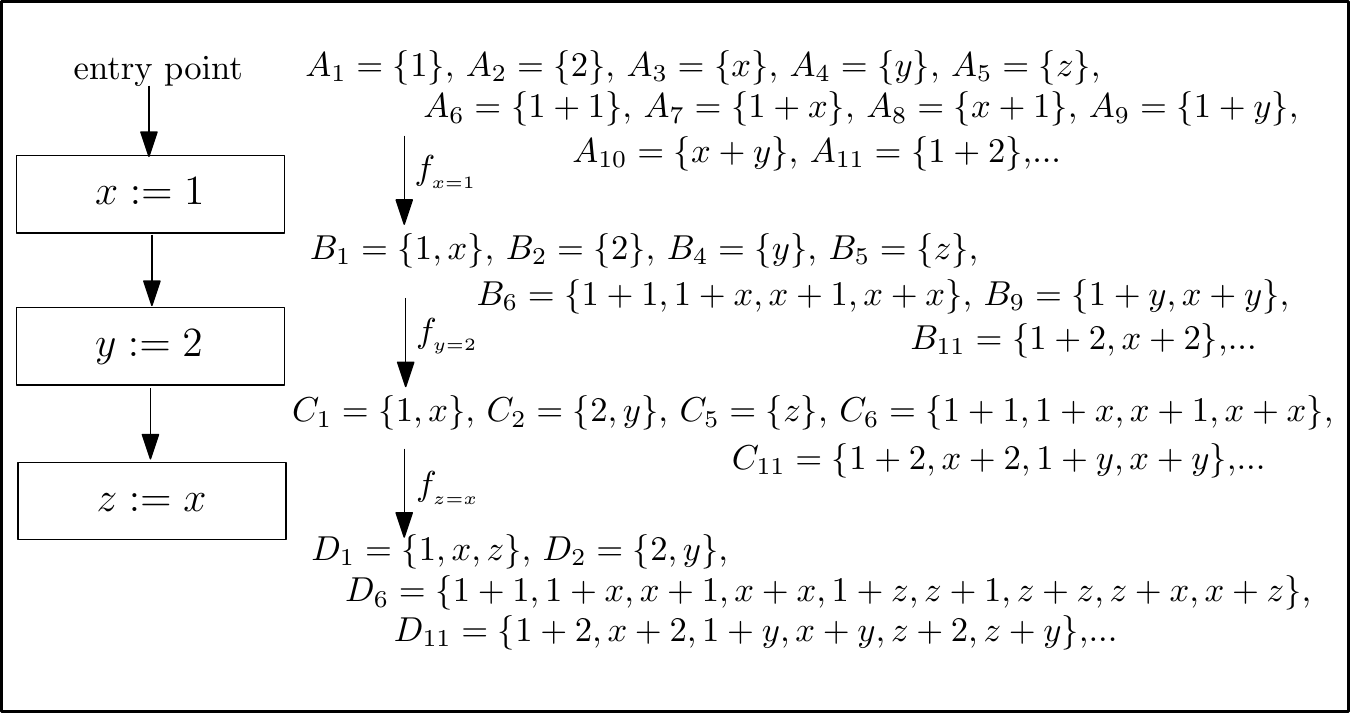}
  \caption{Application of Transfer Functions}
\label{fig:trans} 
\end{center}
\end{figure}
\begin{remark}
It follows from the above definition that $\overline{A_i}(y)=(A_i\setminus A_i(y))\subseteq B_i$.  That is, $B_i$ will contain all terms in $A_i$ in which $y$ does not appear. See Figure~\ref{fig:trans} for an example.
\end{remark}
In the following, we write $f(\mathcal{P})$ instead of $\trans$ to avoid cumbersome notation.  The following is a direct consequence of Definition~\ref{def:trans-function}.
\begin{obs}\label{obs:trans}
 For any $t,t'\in \mathcal{T}$, $t\cong_{f(\mathcal{P})}t'$ if and only if $t[y\leftarrow \beta]\cong_{\mathcal{P}}t'[y\leftarrow \beta]$.  
\end{obs}
To make Definition~\ref{def:trans-function} well founded, we need to establish the following:
\begin{theorem}\label{thm:trans-function}  
 If $\mathcal{P}$ is a congruence, then for any $y\in X$, $\beta \in \overline{\mathcal{T}}(y)$, $\trans$ is a congruence.  
\end{theorem}
Next, we extend the definition of transfer functions to $(\gtbar,\preceq,\bot,\top)$. 
\begin{definition}\label{def:ext-transfun}
Let $y\in X$ and $\beta\in \overline{ \mathcal{T}}(y)$.  Let $\mathcal{P}\in \gtbar$. The extended transfer function 
$\transbar: \gtbar \longrightarrow \gtbar$ transforms $\mathcal{P}$ to $\transbar\in \gtbar$ given by the following:  
\begin{itemize}
 \item If $\mathcal{P}\in \gt$, $\transbar=\trans$.
 \item $\overline{f}_{_{y=\beta}}(\top)=\top$.
\end{itemize}
\end{definition}  
To simplify the notation, we often write $\trans$ (or even simply $f(\mathcal{P})$) 
instead of $\transbar$, and refer to extended transfer functions as simply transfer functions, when the underlying
assignment operation is clear from the context. 
\section{Properties of Transfer Functions}\label{sec:transfun-2}
In this section we show that transfer functions are continuous over the complete lattice $(\gtbar, \preceq, \bot, \top)$.  

Consider the (extended) transfer function $f=\overline{f}_{_{y=\beta}}$, where $y\in X$, $\beta\in \overline{\mathcal{T}}(y)$.  Let $\mathcal{P}_1$ and 
$\mathcal{P}_2$ be congruences in $\gtbar$, not necessarily distinct.     
\begin{lemma}[Distributivity]\label{lem:distrb}
 $\transone \wedge \transtwo = \transmeet$. 
\end{lemma}
Since distributive functions are monotone, we have:  
\begin{corollary}[Monotonicity]\label{cor:monotone}
 If $\mathcal{P}_1\preceq \mathcal{P}_2$, then $\transone \preceq \transtwo$.   
\end{corollary}
We next show that distributivity extends to arbitrary collections of congruences.   
\begin{definition}\label{def:cts}
Let $(L,\leq,\bot,\top)$ and $(L'\leq',\bot',\top')$ be complete lattices. A function $f: L\rightarrow L'$ is continuous 
if for each $\emptyset \neq S\subseteq L$, $f(\bigwedge S)=\bigwedge' f(S)$, where $f(S)=\{f(s): s\in S\}$ and $\bigwedge, \bigwedge'$
denote the infimum operations in the lattices $L$ and $L'$ respectively.  
\end{definition}
\begin{remark}
 The definition of continuity given above is more stringent than the standard definition in the literature, which
 requires continuity only for subsets that are chains. Moreover, the definition above exempts the continuity condition to
 hold for the empty set, because otherwise even constant maps will fail to be continuous.  
\end{remark}
The proof of the next theorem uses the axiom of choice.  
Let $f=f_{_{y=\beta}}$, where $y\in X$, $\beta\in \overline{ \mathcal{T}}(y)$.  For arbitrary collections of congruences $S\subseteq \gtbar$,
The notation $f(S)$ denotes the set $\{f(s): s\in S\}$.  
\begin{theorem}[Continuity]\label{thm:continuous}
 For any $\emptyset \neq S \subseteq \gtbar$, $f(\bigwedge S)=\bigwedge f(S)$.  
\end{theorem}
\section{Non-deterministic assignment}\label{sec:non-det-asg}
Next we define a special kind of transfer functions corresponding to input statements in the program.  This kind of transfer functions are called
non-deterministic assignments.    
\begin{definition}\label{def:non-det-trans}
Let $y\in X$ and let $\mathcal{P}\in \gt$. The transfer function 
$\ntrans: \mathcal{G}(\mathcal{T})\mapsto \mathcal{G}(\mathcal{T})$ transforms $\mathcal{P}$ to a congruence $f(\mathcal{P})=\ntrans$, given by:
for every $t,t'\in \mathcal{T}$, $t\confp t'$ if and only if both the following conditions are satisfied:  
\begin{enumerate}
 \item $t\conp t'$
 \item For every $\beta \in \mathcal{T}\setminus \mathcal{T}(y), t[y\leftarrow \beta]\conp t'[y\leftarrow \beta]$.   
\end{enumerate}
\end{definition}
Since for every pair of terms $t,t'\in \mathcal{T}$ the above definition precisely decides whether $t\confp t'$ or not, $\ntrans$ is the unique relation containing
exactly those pairs $t,t'\in \mathcal{T}$ satisfying both the conditions in Definition~\ref{def:non-det-trans}.  
The definition asserts that two terms that were equivalent before a non-deterministic assignment, will remain equivalent after the non-deterministic assignment to $y$ 
if and only if the equivalence between the two terms is preserved under all possible substitutions to $y$.  

To make the above definition well founded, we need to prove that $\ntrans$ is a congruence.  
\begin{theorem}\label{thm:non-det-trans}
If $\mathcal{P}$ is a congruence, then for any $y\in X$, $\ntrans$ is a congruence.  
\end{theorem}
We write $\bigwedge_{\beta \in \overline{\mathcal{T}}(y)} \trans$  
to denote the set $\bigwedge\{ \trans : \beta \in \overline{\mathcal{T}}(y)\}$.  The next theorem shows that each non-deterministic assignment
may be expressed as a confluence of (an infinite collection of) transfer function operations.   
\begin{theorem}\label{thm:ntrans-char}
If $\mathcal{P}$ is a congruence, then for any $y\in X$, \\
\indent $\ntrans=\mathcal{P}\wedge \left( \bigwedge_{\beta \in \overline{\mathcal{T}}(y)} \trans \right)$.
\end{theorem}
Next, we extend the definition of non-deterministic assignment transfer functions to the complete lattice $(\gtbar,\preceq,\bot,\top)$. 
\begin{definition}\label{def:ext-ntransfun}
Let $y\in X$ and $\mathcal{P}\in \gtbar$. The extended transfer function 
$\ntransbar: \gtbar \mapsto \gtbar$ transforms $\mathcal{P}$ to $\ntransbar\in \gtbar$ given by the following:  
\begin{itemize}
 \item If $\mathcal{P}\in \gt$, $\ntransbar=\ntrans$.
 \item $\overline{f}_{_{y=*}}(\top)=\top$.
\end{itemize}
\end{definition}  
The following theorem involves use of the axiom of choice.  We will write $\ntrans$ instead of $\ntransbar$ to simplify notation.  
\begin{theorem}[Continuity]\label{thm:ntrans-cts}
 For any $\emptyset \neq S \subseteq \gtbar$, ${f}_{_{y=*}}(\bigwedge S)=\bigwedge {f}_{_{y=*}}(S)$, where\\  ${f}_{_{y=*}}(S)=\{f_{_{y=*}}(s): s\in S\}$.
\end{theorem}
In the following, we derive a characterization for non-deterministic assignment that does not depend on the axiom of choice.   Condition (3) of the definition of congruence
(Definition~\ref{def:cong}) is necessary to derive this characterization.  We first note a lemma which states that if the equivalence 
between two terms $t,t'$ is preserved under substitution of $y$ with any two distinct constants chosen arbitrarily, 
then the equivalence between the two terms will be preserved under substitution of $y$  with any other term $\beta$ in which $y$ does not appear.   
\begin{lemma}\label{lem:nondet-constant}
 Let $\mathcal{P}\in \gt$.  Let $t,t'\in \mathcal{T}$ and $c_1,c_2\in C$ with $c_1\neq c_2$.  Then 
 $t[y\leftarrow c_1]\conp t'[y\leftarrow c_1]$ and $t[y\leftarrow c_2]\conp t'[y\leftarrow c_2]$ if and only if  
 $t[y\leftarrow \beta]\conp t'[y\leftarrow \beta]$ for all $\beta\in \overline{\mathcal{T}}(y)$.  
\end{lemma}
The above observation leads to a characterization of non-deterministic assignment that does not involve the axiom of choice.   
\begin{theorem}\label{thm:nondet-const}
Let $\mathcal{P}\in \gt$ and let $c_1,c_2\in C$, $c_1\neq c_2$.  Then, for any $y\in X$,\\ 
$\ntrans=\mathcal{P}\wedge f_{_{y\leftarrow c_1}}(\mathcal{P}) \wedge f_{_{y\leftarrow c_2}}(\mathcal{P})$
\end{theorem}
It follows from the above theorem that non-deterministic assignments can be characterized in terms of just three 
congruences (instead of dealing with infinitely many as in Theorem~\ref{thm:ntrans-char}). 
\section{Data Flow Analysis Frameworks}\label{sec:data-flow-frameworks}
We next formalize the notion of a data flow framework and apply the formalism developed above to characterize Herbrand equivalence at each point in 
a program.   
\begin{definition}\label{def:control-flow-graph}
 A control flow graph $G(V,E)$ is a directed graph over the vertex set $V=\{1,2,\ldots,n\}$ for some  $n\geq 1$ satisfying the following properties:  
 \begin{itemize}
  \item $1\in V$ is called the entry point and has no predecessors.
  \item Every vertex $i\in V$, $i\neq 1$ is reachable from $1$ and has at least one predecessor and at most two predecessors.
  \item Vertices with two predecessors are called confluence points.
  \item Vertices with a single predecessor are called (transfer) function points.
 \end{itemize}
\end{definition}
\begin{definition}\label{def:data-flow-framework}
A data flow framework over $\mathcal{T}$ is a pair $D=(G,F)$, where  $G(V,E)$ is a control flow graph on the vertex set $V=\{1,2,\ldots, n\}$ and 
 $F$ is a collection of transfer functions over $\gt$ such that for each function point $i\in V$, there is an associated transfer function
$h_i\in F$, and $F=\{h_i: i\in V \text{ is a transfer function point}\}$.   
\end{definition}
\vspace{0.1cm}
Data flow frameworks can be used to represent programs.  An example is given in Figure~\ref{fig:programpoints}.
\begin{remark}
In the sections that follow, for any $h_i\in F$, we will simply write $h_i$ to actually denote the extended transfer function $\overline{h}_i$ 
(see Definition~\ref{def:ext-transfun} and Definition~\ref{def:ext-ntransfun}) without further explanation.  
\end{remark}
\section{Herbrand Equivalence}\label{sec:Herbrand}
Let $D=(G,F)$ be a data flow framework over $\mathcal{T}$.  In the following, we will define the Herbrand Congruence function $H_D : V(G)\mapsto \gtbar$.  
For each vertex $i\in V(G)$, the congruence $H_D(i)$ will be called
the \emph{Herbrand Congruence} associated with the vertex $i$ of the data flow framework $D$.  The function $H_D$ will be defined as the maximum fix-point of a continuous
function $\fd:\gtbar^{n}\mapsto \gtbar^{n}$. The function $\fd$ will be called the \emph{composite transfer function} associated with the data flow
framework $D$. 
\begin{definition}[Product Lattice]\label{def:Prod_lattice}  
Let $n$ a positive integer.  The product lattice, \\$(\gtbar^{n},\preceq_n,\bot_n,\top_n)$ is defined as follows: 
for $\overline{\mathcal{P}}=(\mathcal{P}_1, \mathcal{P}_2,\ldots, \mathcal{P}_n)$, 
$\overline{\mathcal{Q}}=(\mathcal{Q}_1, \mathcal{Q}_2,\ldots, \mathcal{Q}_n) \in \gtbar^{n}$, $\overline{\mathcal{P}}\preceq_n \overline{\mathcal{Q}}$ 
if $\mathcal{P}_i\preceq \mathcal{Q}_i$ for each $1\leq i\leq n$, 
$\bot_n=(\bot, \bot, \ldots,\bot)$ and $\top_n=(\top,\top, \ldots, \top)$.
\end{definition}
For $S\subset \gtbar^{n}$, the notation $\bigwedge_n S$ will be used to denote the least upper bound of $S$ in the product lattice.  

By Theorem~\ref{thm:complete}, Theorem~\ref{Athm:prod_complete} and Corollary~\ref{Acor:prod_complete}, we have:
\begin{theorem}\label{thm:prod_complete}
The product lattice satisfies the following properties:  
\begin{enumerate}
 \item $(\gtbar^{n},\preceq_n,\bot_n,\top_n)$ is a complete lattice.
 \item If $\tilde{S}\subseteq \gtbar^{n}$ is non-empty, with $\tilde{S}=S_1\times S_2\times \cdots \times S_n$,  where $S_i\subseteq \gtbar$ for $1\leq i\leq n$.  
       Then $\bigwedge_n \tilde{S}=(\bigwedge S_1,\bigwedge S_2, \ldots, \bigwedge S_n)$.
\end{enumerate}
\end{theorem}
As preparation for defining the composite transfer function, we introduce the following functions: 
\begin{definition}[Projection Maps]\label{def:proj_map} Let $n$ be a positive integer. For each $i\in \{1,2,\ldots,n\}$,
\begin{itemize}
 \item The projection map to the $i^{th}$ co-ordinate, 
       $\pi_i:\gtbar^{n}\mapsto \gtbar$ is defined by \\$\pi_i(\mathcal{P}_1,\mathcal{P}_2,\ldots,\mathcal{P}_n)=\mathcal{P}_i$ for any 
       $(\mathcal{P}_1,\mathcal{P}_2,\ldots,\mathcal{P}_n)\in \gtbar^{n}$.  
 \item The confluence map $\pi_{i,j}:\gtbar^{n}\mapsto \gtbar$ is defined by\\
       $\pi_{i,j}(\mathcal{P}_1,\mathcal{P}_2,\ldots, \mathcal{P}_n)=\mathcal{P}_i\wedge \mathcal{P}_j$ for any $(\mathcal{P}_1,\mathcal{P}_2,\ldots, \mathcal{P}_n)\in \gtbar^{n}$.
\end{itemize}
\end{definition}
In addition to the above functions, we will also use the constant map which maps each element in $\gtbar^{n}$ to $\bot$.  
The following observation is a consequence of Theorem~\ref{Athm:conf_cts}.
\begin{obs}\label{obs:conf_cts}
Constant maps, projection maps and confluence maps  are continuous.
\end{obs}
For each $k\in V(G)$, $\pred(k)$ denotes the set of predecessors of the vertex $k$ in the control flow graph $G$.
\begin{definition}[Composite Transfer Function]\label{def:comp-trans-fun}
Let $D=(G,F)$ be a data flow framework over $\mathcal{T}$. For each $k\in V(G)$, define the component map $f_k:\gtbar^{n}\mapsto \gtbar$ as follows:  
\begin{enumerate}
 \item If $k=1$, the entry point, then $f_k=\bot$.  ($f_1$ is the constant function that always returns the value $\bot$). 
 \item If $k$ is a function point with $\pred(k)=\{j\}$,  then $f_k=h_k\circ \pi_j$, where $h_k$ is the (extended) transfer function corresponding to the function point $k$, 
       and $\pi_k$ the projection map to the  $k^{th}$ coordinate as defined in Definition~\ref{def:proj_map}. 
 \item If $k$ is a confluence point with $\pred(k)=\{i,j\}$, then $f_k=\pi_{i,j}$, where $\pi_{i,j}$ is the confluence map as defined in Definition~\ref{def:proj_map}.  
\end{enumerate}
The composite transfer function of $D$ is defined to be the unique function (Observation~\ref{Aobs:component}) $\fd$ satisfying $\pi_k\circ \fd=f_k$ for each $k\in V(G)$.   
\end{definition}
The purpose of defining $\fd$ is the following.  
Suppose we have associated a congruence with each program point in a data flow framework. Then $\fd$ specifies 
how a simultaneous and synchronous application of all the transfer functions/confluence operations at 
the respective program points modifies the congruences at each program point.  The definition of $\fd$ conservatively 
sets the confluence at the entry point to $\bot$, treating terms in $\gt$ to be inequivalent to each other at the entry point. 
See Figure~\ref{fig:programpoints} for an example. The following observation is a direct consequence of the above definition.
\begin{obs}\label{obs:comp-trans}
The composite transfer function $\fd$ (Definition~\ref{def:comp-trans-fun}) satisfies the following properties:
\begin{enumerate}
 \item If $k=1$, the entry point, then $\pi_k\circ \fd =\bot$.  
 \item If $k$ is a function point with $\pred(k)=\{j\}$, then $f_k=\pi_k\circ \fd = h_k\circ \pi_j$,   where $h_k$ is the (extended) transfer function corresponding to the function point $k$.  
 \item If $k$ is a confluence point with $\pred(k)=\{i,j\}$, then $f_k=\pi_k \circ \fd = \pi_{i,j}$.  
\end{enumerate}
\end{obs}
The following lemma is a consequence of Observation~\ref{obs:comp-trans}.
\begin{lemma}\label{lem:fd-induction} 
Let $D=(G,F)$ be a data flow framework over $\mathcal{T}$ and  $k\in V(G)$.  Let $S=\{\fd(\top_n), \fd^2(\top_n),\ldots \}$, where $\fd$ is the
composite transfer function of $D$.  
\begin{enumerate}
 \item If $k=1$, the entry point, then $\pi_k\circ \fd^{l}(\top_n)=\bot$ for all $l\geq 1$, hence $\pi_k (\bigwedge_n S) =\top$.  
 \item If $k$ is a function point with $\pred(k)=\{j\}$,  then for all $l\geq 1$, 
 \begin{eqnarray*}
          (\pi_k \circ \fd^{l})(\top_n)&=& (\pi_k\circ \fd)(\fd^{l-1}(\top_n))\\
                                       &=&(h_k\circ \pi_j \circ \fd^{l-1})(\top_n)
 \end{eqnarray*}
 \item If $k$ is a confluence point with $\pred(k)=\{i,j\}$, then for all $l\geq 1$,
\begin{eqnarray*}
 (\pi_k \circ \fd^{l})(\top_n)&=&(\pi_k\circ \fd)(\fd^{l-1}(\top_n)) \\
                              &=& (\pi_{i, j})(\fd^{l-1}(\top_n))\\
                              &=& \left((\pi_i \circ \fd^{l-1})(\top_n)\right)  \wedge \left((\pi_j \circ \fd^{l-1})(\top_n)\right)  
\end{eqnarray*}
\end{enumerate}
\end{lemma}
By Theorem~\ref{Athm:prod_cts}, Observation~\ref{Aobs:monotone} and Corollary~\ref{Acor:MFP}, we have:  
\begin{theorem}\label{thm:prop-composite-trans}
The following properties hold for the composite transfer function $\fd$ (Definition~\ref{def:comp-trans-fun}):  
\begin{enumerate}
 \item $\fd$ is monotone, distributive and continuous.  
 \item The component maps $f_k=\pi_k\circ \fd$ are continuous for all $k\in \{1,2,\ldots, n\}$.  
 \item $\fd$ has a maximum fix-point.  
 \item If $S=\{\top, \fd(\top_n), \fd^2(\top_n),\ldots \}$, then $\bigwedge_n S$ is the maximum fix-point of $\fd$.  
\end{enumerate}
\end{theorem}
 \begin{figure}[h] 
  \begin{center}
 \includegraphics[scale=0.85]{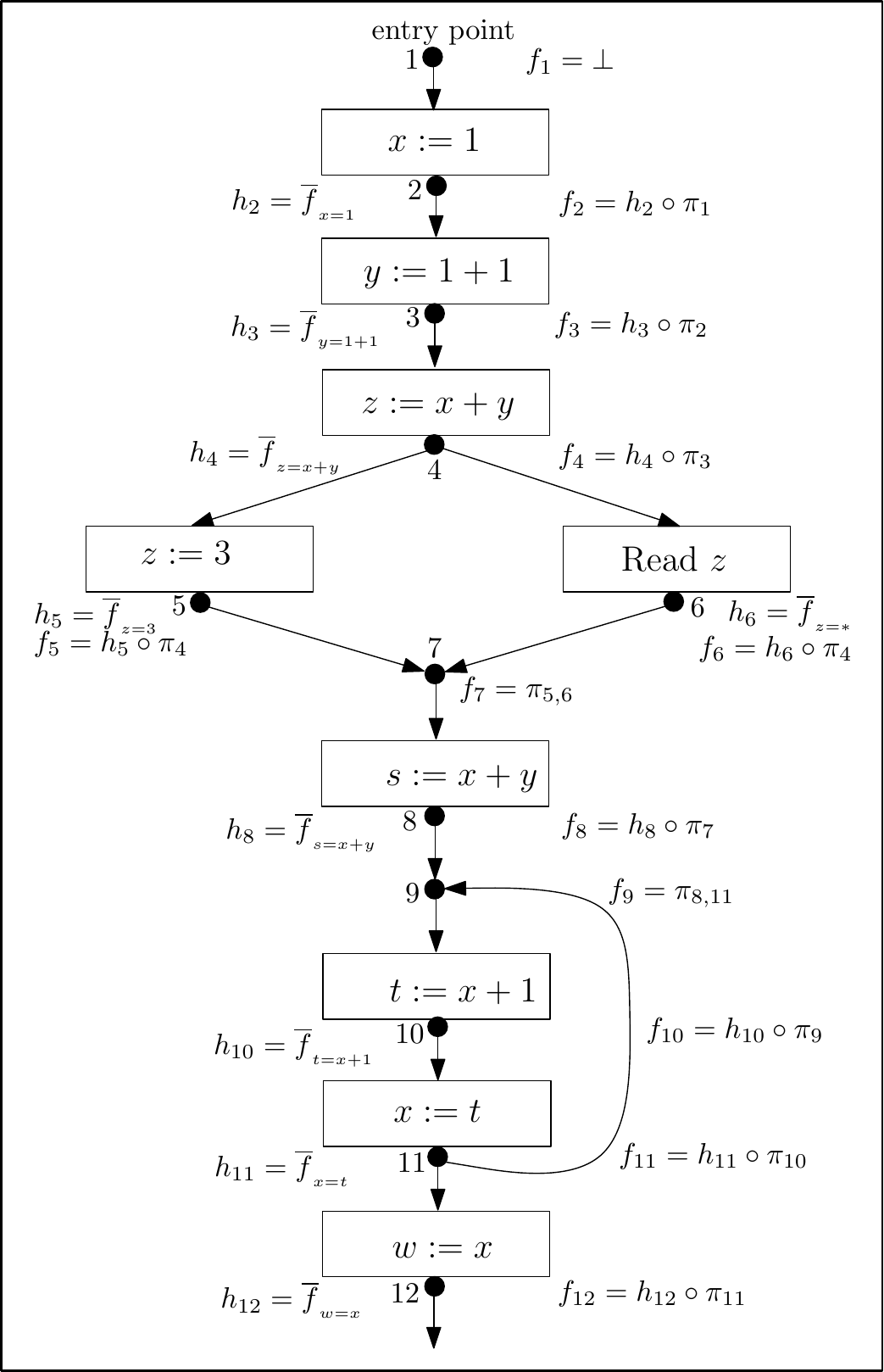}
  \caption{Component Maps of the Composite Transfer Function}
\label{fig:programpoints} 
\end{center}
\end{figure}
The objective of defining Herbrand Congruence as the maximum fix point of the composite transfer function is possible now.    
\begin{definition}[Herbrand Congruence]\label{def:Herbrand}
Given a data flow framework $D=(G,F)$ over $\mathcal{T}$, the Herbrand Congruence function  $H_D : V(G)\mapsto \gtbar$ is defined as the maximum fix point of the 
composite transfer function $\fd$.   For each $k\in V(G)$, the value $H_D(k)\in \gtbar$ is referred to as the Herbrand Congruence at program point $k$.  
\end{definition}
The following is a consequence of Theorem~\ref{thm:prop-composite-trans} and the definition of Herbrand Congruence.  
\begin{obs}\label{obs:herbrand}
For each $k\in V(G)$, $H_D(k)=\bigwedge \{ (\pi_k \circ {\fd}^{l})(\top_n) :  l\geq 0 \}$.  
\end{obs}
\begin{proof}
 \begin{eqnarray*}
  H_D(k)&=&\pi_k(\bigwedge\text{$_{_n}$} \{{\fd}^{l}(\top_n) : l\geq 0\}) \text{ (by Theorem~\ref{thm:prop-composite-trans}) }\\
        &=&\bigwedge \{ \pi_k({\fd}^{l}(\top_n)) : l\geq 0 \} \text{ (by continuity of $\pi_k$) } \\
        &=&\bigwedge \{ (\pi_k\circ {\fd}^{l})(\top_n) : l\geq 0\}
 \end{eqnarray*}
\end{proof}
The definition of Herbrand congruence must be shown to be consistent with the traditional meet-over-all-paths description of Herbrand equivalence of terms in a data flow
framework. The next section addresses this issue.  
\section{MOP characterization}\label{sec:MOP}
In this section, we give a meet over all paths characterization for the Herbrand Congruence at each program point. 
This is essentially a lattice theoretic reformulation of the characterization presented by Steffen et.~al.~\cite[p.~393]{Steffen90}.
In the following, assume that we are given a data flow framework 
$D=(G,F)$ over $\mathcal{T}$, with $V(G)=\{1,2,\ldots, n\}$.    
\begin{definition}[Path]\label{def:prog-path}
For any non-negative integer $l$, a program path (or simply a path) of length $l$ to a vertex $k\in V(G)$ is a sequence $\alpha=(v_0,v_2,\ldots v_{l})$ 
satisfying $v_0=1$, $v_l=k$ and $(v_{i-1},v_{i})\in E(G)$ for 
each $i\in \{1,2,\ldots l\}$. For each $i\in\{0,1,\ldots, l\}$,  $\alpha_i$ denotes the initial segment of $\alpha$ up to the  $i^{th}$ vertex, given by $(v_0,v_1,\ldots, v_i)$.
\end{definition}
\begin{remark}
Note that the vertices in a path need not be distinct under this definition.   
\end{remark}
The next definition associates a congruence in $\gtbar$ with each path in $D$.  The path function captures the effect of application of transfer functions along the path on 
the initial congruence $\bot$, in the order in which the transfer functions appear along the path.  
 \begin{definition}[Path Congruence] \label{def:path-fun}
Let $\alpha=(v_0,v_1,\ldots, v_l)$ be a path of length $l$ to vertex $k\in V(G)$ for some $l\geq 0$. We define:  
\begin{enumerate}
 \item When $i=0$, $m_{\alpha_i}=\bot$.
 \item If $i>0$ and $v_{i}=j$, where $j\in V(G)$ is a function point, then $m_{\alpha_i}=h_j(m_{\alpha_{i-1}})$, where $h_j\in F$
       is the extended transfer function associated with the function point $j$. 
 \item If $i>0$ and $v_{i}$ is a confluence point, then $m_{\alpha_i}=m_{\alpha_{i-1}}$.  
 \item $m_\alpha=m_{\alpha_l}$.
\end{enumerate}
The congruence $m_\alpha$ is defined as the path congruence associated with the path $\alpha$.    
\end{definition}
For $k\in V(G)$ and $l\geq 0$, let $\Phi_l(k)$ denote the set of all paths of length {\em less than $l$} from the entry point $1$ to the vertex $k$.  
In particular, $\Phi_0(k) =\emptyset$, for all $k \in V(G)$.
The following observation is a consequence of the definition of  $\Phi_l(k)$.
\begin{obs}\label{obs:obs-path}
 If $k \in V(G)$ and $l\geq 1$, 
 \begin{enumerate}
   \item If $k$ is the entry point, then $\Phi_{l}(k)=\{(1)\}$, the set containing only the path of length zero, starting and ending at vertex $1$.
   \item If $k$ is a function point with $\pred(k)=\{j\}$, then \\
   $\{\alpha_{l-1} : \alpha\in \Phi_l(k)\}=\{\alpha' : \alpha' \in  \Phi_{l-1}(j)\}=\Phi_{l-1}(j)$.
   \item If $k$ is a confluence point with $\pred(k)=\{i,j\}$, then\\ 
   $\{\alpha_{l-1} : \alpha\in \Phi_l(k)\}=\{\alpha' : \alpha' \in  \Phi_{l-1}(i)\} \cup \{\alpha' : \alpha' \in  \Phi_{l-1}(j)\}=\Phi_{l-1}(i) \cup \Phi_{l-1}(j)$.
 \end{enumerate}
\end{obs}
For $l\ge 0$, we define the congruence $M_l(k)$ to be the meet of all path congruences associated with paths of length less than $l$ from the entry point to vertex $k$ in $G$. 
Stated formally,
$$ M_l(k)=\bigwedge\{ m_\alpha: \alpha \in \Phi_{l}(k)\}, \text{for $l \ge 0$}$$
\begin{obs}\label{obs:base-cases}
If $l=0$, $\Phi_{l}(k)=\emptyset$ and hence $M_0(k)=\top$, for all $k \in V(G)$. Further, $M_1(1)=\bot$ and $M_1(k)=\top$, for $k\ne 1$. In general, 
$M_l(k)=\top$ if there are no paths of length less than $l$ from $1$ to $k$ in $G$.  
\end{obs}
We define $\Phi_k$ to be the set of all paths from vertex $1$ to vertex $k$ in $G$, i.e., $\Phi_k=\bigcup_{l\geq 1} \Phi_l(k)$  
and $MOP(k)=\bigwedge\{m_\alpha: \alpha \in \Phi(k)   \}=\bigwedge\{ M_l(k): l\geq 0\}$.  
(The second equality follows from Lemma~\ref{Alem:meet-union} and Observation~\ref{obs:base-cases}.) 
The congruence $MOP(k)$ is the meet of all path congruences associated with paths in $\Phi_k$. 

Our objective is to prove $MOP(k)=H_{D}(k)$ for each $k\in \{1,2,\ldots, n\}$ so that $H_D$ captures the meet over all paths information about equivalence
of expressions in $\mathcal{T}$.  We begin with the following observations.  
\begin{lemma}\label{lem:paths}
 For each $k\in V(G)$ and $l\geq 1$
  \begin{enumerate}
 \item If $k=1$, the entry point, then $M_l(k)=\bot$.  
 \item If $k$ is a function point with $\pred(k)=\{j\}$, then $M_l(k)=h_k(M_{l-1}(j))$,  where $h_k$ is the (extended) transfer function corresponding to the function point $k$.  
 \item If $k$ is a confluence point with $\pred(k)=\{i,j\}$, then $M_l(k)=M_{l-1}(i) \wedge M_{l-1}(j)$.  
\end{enumerate}
\end{lemma}
\begin{proof}
 If $k=1$, $\Phi_l(k)=\{(1)\}$, by Observation~\ref{obs:obs-path}.  Hence $M_l(k)=M_{1}(1)=\bot$, by Observation~\ref{obs:base-cases}.\\
 If $k$ is a function point with $\pred(k)=\{j\}$, then
 \begin{eqnarray*}
   M_l(k)&=&\bigwedge\{ m_\alpha: \alpha \in \Phi_l(k)\}\\
   &=&\bigwedge\{ h_k (m_{\alpha_{l-1}}): \alpha \in \Phi_l(k)\}\text{ (by Definition~\ref{def:path-fun})} \\
   &=&\bigwedge\{ h_k (m_{\alpha'}): \alpha' \in \Phi_{l-1}(j)\}\text{ (by Observation~\ref{obs:obs-path})}\\ 
   &=&h_k\left(\bigwedge\{m_{\alpha'}: \alpha' \in \Phi_{l-1}(j)\}\right)\text{ (by continuity of $h_k$) }\\ 
   &=&h_k(M_{l-1}(j))\text{ (by definition of $M_{l-1}(j)$)}
 \end{eqnarray*}
If $k$ is a confluence point with $\pred(k)=\{i,j\}$, then 
 \begin{eqnarray*}
  M_l(k)&=&\bigwedge\{ m_\alpha: \alpha \in \Phi_l(k)\}\\
        &=&\bigwedge( \{m_{\alpha_{l-1}}: \alpha \in \Phi_l(k)\}) \text{ (by Definition~\ref{def:path-fun})}\\
        &=&\bigwedge\left(\{m_{\alpha'}: \alpha' \in \Phi_{l-1}(i)\} \cup \{m_{\alpha'}: \alpha' \in \Phi_{l-1}(j) \}\right) \text{ (by Observation~\ref{obs:obs-path})}  \\
         &=&\left( \bigwedge \{m_{\alpha'}: \alpha' \in \Phi_{l-1}(i)\}\right) \wedge \left(\bigwedge\{m_{\alpha'}: \alpha' \in \Phi_{l-1}(j)\}\right)\text{ (by Lemma~\ref{Alem:meet-union})}\\
        &=&M_{l-1}(i) \wedge M_{l-1}(j)\text{ (by definition of $M_{l-1}(i)$ and $M_{l-1}(j)$)} 
 \end{eqnarray*}
\end{proof}
\begin{lemma}\label{lem:MOP-MFP}
For each $k\in V(G)$ and $l\geq 0$, $M_l(k)=(\pi_k \circ {\fd}^{l})(\top_n)$.
\end{lemma}
\begin{proof}
 Let $k \in V(G)$ be chosen arbitrarily.  We prove the lemma by induction on $l$.\\
 When $l=0$, $M_l(k)=\top$, by Observation~\ref{obs:base-cases}, as required. Otherwise,
 \begin{enumerate}
  \item If $k=1$, the entry point, then  by Lemma~\ref{lem:fd-induction} and  Lemma~\ref{lem:paths},\\ $M_l(k)=\bot=(\pi_k \circ {\fd}^{l})(\top_n)$. 
  \item If $k$ is a function point with $\pred(k)=\{j\}$ and $h_k$ is the (extended) transfer function corresponding to the function point $k$, then
  \begin{eqnarray*}
   M_l(k)&=&h_k(M_{l-1}(j)) \text{ (By Lemma~\ref{lem:paths})} \\
         &=&h_k((\pi_j \circ {\fd}^{l-1})(\top_n)) \text{ (By induction hypothesis) } \\
         &=&(h_k\circ \pi_j\circ {\fd}^{l-1})(\top_n)  \\
         &=&(\pi_k \circ \fd^{l})(\top_n) \text{ (By Lemma~\ref{lem:fd-induction}) }         
  \end{eqnarray*}
  \item If $k$ is a confluence point with $\pred(k)=\{i,j\}$, then
  \begin{eqnarray*}
   M_l(k)&=& M_{l-1}(i) \wedge M_{l-1}(j) \text{ (By Lemma~\ref{lem:paths}) }\\
         &=& \left((\pi_i \circ {\fd}^{l-1})(\top_n)\right) \wedge  \left((\pi_j \circ {\fd}^{l-1})(\top_n)\right) \text{ (By induction hypothesis) } \\
         &=& (\pi_k \circ \fd^{l})(\top_n) \text{ (By Lemma~\ref{lem:fd-induction}) }          
  \end{eqnarray*}
 \end{enumerate}
\end{proof}
Finally, we show that the iterative fix-point characterization of Herbrand equivalence and the meet over all paths characterization coincide.  
\begin{theorem}\label{thm:MOP-MFP}
Let $D=(G,F)$ be a data flow framework.  Then, for each $k\in V(G)$, $MOP(k)=H_D(k)$.  
\end{theorem}
\begin{proof}
\begin{eqnarray*}
   MOP(k)&=&\bigwedge\{m_\alpha: \alpha \in \Phi(k) \}  \\
         &=&\bigwedge\{ M_l(k): l\geq 0\}  \text{ ( by Lemma~\ref{Alem:meet-union} and Observation~\ref{obs:base-cases}) }\\
         &=&\bigwedge\{ (\pi_k \circ {\fd}^{l})(\top_n) : l\geq 0 \} \text{ (by Lemma~\ref{lem:MOP-MFP}) }\\       
         &=& H_D(k)  \text{ (by Observation~\ref{obs:herbrand})}
\end{eqnarray*}         
\end{proof}
\section{Conclusion}
We have shown that Herbrand equivalence of terms in a data flow framework admits a lattice theoretic fix-point characterization.  
Though not the concern addressed here, we note that this fix-point characterization naturally leads to algorithms that work on the  
restriction of congruences to terms that actually appear in given data flow framework 
and iteratively compute the maximum fix-point (see, for example, the algorithm presented in Appendix~\ref{sec:appendixc}).  This allows the algorithmic detection of
Herbrand equivalence among the expressions that actually appear in any given program. Moreover, as mentioned in the introduction, analysis of the completeness of iterative fix point based algorithms that detects
Herbrand equivalence among program expressions would be potentially easier using the fix-point characterization of Herbrand equivalences presented here, 
when compared to the current practice of using a meet over all paths characterization as the reference point.
Further, we have shown that the lattice theoretic fix-point characterization presented here is equivalent to a meet over all paths characterization.
\newpage

\appendix
\makeatletter
\edef\thetheorem{\expandafter\noexpand\thesection\@thmcountersep\@thmcounter{theorem}}
\makeatother
\newpage
\section{Proofs of Theorems}
\textbf{Theorem~\ref{thm:confluence}. }
If $\mathcal{P}_1$ and $\mathcal{P}_2$ are congruences, then $\mathcal{P}_1\wedge \mathcal{P}_2$ is a congruence.  
\begin{proof}
Let $\mathcal{P}_1=\{A_i\}_{i\in I}$ and $\mathcal{P}_2=\{B_j\}_{j\in J}$. Let $C_{i,j}=A_i\cap B_j$ for all $i\in I$ and $j\in J$.   Clearly 
$\mathcal{P}_1\wedge \mathcal{P}_2=\{C_{i,j}: i\in I,j\in J, C_{i,j}\neq \emptyset\}$ is a partition of $\mathcal{T}$.  It suffices to prove that 
$\mathcal{P}_1\wedge \mathcal{P}_2$ satisfies properties (1) to (3) of Definition~\ref{def:cong}.  

\begin{enumerate}
 \item [(1)] For any $c,c'\in C$,  $c\concap c' \iff c \conone c'$ and $c\contwo c'\iff c=c'$. 
 \item [(2)] For $t,t',s,s'\in \mathcal{T}$, $t+s \concap t'+s'$ if and only if $t+s \conone t'+s'$ and $t+s \contwo t'+s'$\\
 $\iff t'\conone t$, $t'\contwo t$, $s'\conone s$ and $s'\contwo s$
 $\iff t'\concap t$ and $s'\concap s$.
 \item[(3)] For any $c\in C$, $t\in \mathcal{T}$, $c\concap t$ if and only if $c\conone t$ and $c\contwo t$, only if either $t=c$ or $t\in X$ (by condition (3) of the definition of
 congruence).  
\end{enumerate}
\end{proof}
\textbf{Lemma~\ref{lem:poset}. }
$(\mathcal{G}(\mathcal{T}),\preceq)$ is a meet semi-lattice with meet operation $\wedge$ and bottom element $\bot$. 
\begin{proof}
It is evident from the definition of $\preceq$ that $(\mathcal{G}(\mathcal{T}),\preceq)$ is a partial order.  Let $\mathcal{P}_1=\{A_i\}_{i\in I}$ and $\mathcal{P}_2=\{B_j\}_{j\in J}$ 
be congruences in $\in \mathcal{G}(\mathcal{T})$. We next show that $\mathcal{P}_1\wedge \mathcal{P}_2$ is the meet of $\mathcal{P}_1$ and $\mathcal{P}_2$ with respect to $\preceq$.  
By Definition~\ref{def:confluence}, the relations $\mathcal{P}_1\wedge \mathcal{P}_2\preceq \mathcal{P}_1$ and $\mathcal{P}_1\wedge \mathcal{P}_2\preceq \mathcal{P}_2$ must hold true.  
Suppose that a congruence $\mathcal{P}=\{C_k\}_{k\in K}$ satisfies $\mathcal{P}\preceq \mathcal{P}_1$   and $\mathcal{P}\preceq \mathcal{P}_2$.  Then by definition of $\preceq$,
for each $k\in K$, there must exist $i\in I$ and $j\in J$ such that $C_k\subseteq A_i$ and $C_k\subseteq B_j$.
Consequently,  $\emptyset \neq C_k\subseteq (A_i\cap B_j)$.  By Definition~\ref{def:confluence} we have $A_i\cap B_j\in \mathcal{P}_1\wedge \mathcal{P}_2$ and 
hence $\mathcal{P}\preceq \mathcal{P}_1\wedge \mathcal{P}_2$.  Thus, $\mathcal{P}_1\wedge \mathcal{P}_2$ is the meet of $\mathcal{P}_1$ and $\mathcal{P}_2$.  Finally, $\bot$ is
the bottom element of $\gt$ by Observation~\ref{obs:bot}.
\end{proof}
\textbf{Lemma~\ref{lem:subsetInfimum}. }
 Every non-empty subset of $(\mathcal{G}(\mathcal{T}),\preceq)$ has a greatest lower bound.  
\begin{proof}
 We will show that any arbitrary non-empty family of congruences $\{\mathcal{P}_r\}_{r\in R}$ of $\mathcal{T}$ has a greatest lower bound. 
 For each $r \in R$, let $\mathcal{P}_r =\{A _{r, i_r}\}_{i_r \in I_r}$, where $I_r\neq \emptyset$. Let $\mathcal{I}$ be the Cartesian product of the
 index sets of the congruences; i.e., $\mathcal{I}=\prod_{r \in R}{I_r}$. Note that, each element in $\mathcal{I}$ is a sequence $(i_r)_{r \in R}$.
 For each $(i_r)_{r\in R} \in \mathcal{I}$, define $B_{(i_r)}= \bigcap_{r\in R} A_{r,i_r}$.   (Note that the index element $i_r \in I_r$ and $A_{r,i_r}$ is the set in partition  
 $\{\mathcal{P}_r\}_{r\in R}$ having index $i_r$). 
 Let $\mathcal{P}=\{B_{(i_r)} : (i_r) \in \mathcal{I}, B_{(i_r)} \ne \emptyset \}$. 
 
 We claim that $\mathcal{P}$ is a congruence and it is the greatest lower bound of  $\{\mathcal{P}_r\}_{r\in R}$.  To show that $\mathcal{P}$ is a partition, first assume $t\in \mathcal{T}$.
 Then, for each $r\in R$, there exists some $i_r\in I_r$ such that $t\in A_{r,i_r}$.  Hence $t\in B_{(i_r)}$.  Thus every element appears in at least one set $B_{(i_r)}$ in $\mathcal{P}$.  
 Next, if $(i_r),(j_r)\in \mathcal{I}$ such that $(i_r)\neq (j_r)$,  we show that $B_{(i_r)}\cap B_{(j_r)}=\emptyset$.  To this end, assume that some $t\in \mathcal{T}$ satisfies 
 $t\in B_{(i_r)}\cap B_{(j_r)}$.  Since $(i_r)\neq (j_r)$, there exists some index $r_0$ in which the sequences $(i_r)$ and $(j_r)$ differ; i.e, 
 $i_{r_0}\neq j_{r_0}$.  As $t\in B_{(i_r)}\cap B_{(j_r)}$, it must be the case that $t\in A_{r_0, i_{r_0}}\cap A_{r_0, j_{r_0}}$,
 which contradicts our assumption that $\mathcal{P}_{r_0}$ is a congruence.  

 Next, we prove that $\mathcal{P}$ satisfies properties (1) to (3) of Definition~\ref{def:cong}. 
 \begin{enumerate}
  \item [(1)] Let $c,c'\in C$ be constants such that $c\neq c'$.    Suppose $c,c'\in B_{(i_r)}$ for some $(i_r)\in \mathcal{I}$.  
    Then, for each $r\in R$, $c,c'\in A_{r,i_r}$.  However, this is impossible as $\{A_{r,i_r}\}_{i_r\in I_r}$ is a congruence for each $r$.  This proves (1).  
  \item [(2)] Let $t,t',s,s'\in \mathcal{T}$.  We have $t,t'\in  B_{(i_r)}$ and $s,s'\in  B_{(j_r)}$ for some $(i_r),(j_r)\in \mathcal{I}$ if and only if 
  $t,t'\in A_{r,i_r}$ and $s,s'\in A_{r,j_r}$ for each $r\in R$ if and only if there exists $(k_r)\in \mathcal{I}$ such that $t+s,t'+s'\in A_{r,k_r}$ for each $r\in R$ (because each
  $\mathcal{P}_r$ is a congruence) if and only if $t+s,t'+s'\in B_{(k_r)}$ (by definition of $B_{(k_r)}$), proving (3).
  \item [(3)] For any $c\in C$ and $t\in \mathcal{T}$, suppose $c,t\in B_{(i_r)}$.   Then, $c,t\in A_{r,i_r}$ for each $r\in R$.  Hence either $t=c$ or $t\in X$ because $\mathcal{P}_r$ 
  is a congruence for each $r\in R$.  
 \end{enumerate}
  
  Next we show that $\mathcal{P}\preceq \mathcal{P}_r$ for each $r\in R$.   Let $B_{(i_r)}\in \mathcal{P}$.  By definition, $B_{(i_r)}=\bigcap_{r\in R} A_{r,i_r}$. Thus $B_{(i_r)}\subseteq A_{r,i_r}$ for each $r\in R$.    
  As $A_{r,i_r}\in \mathcal{P}_r$ for each $r\in R$, we have $\mathcal{P}\preceq \mathcal{P}_r$ for each $r\in R$.
 
  To prove that $\mathcal{P}$ is the greatest lower bound of $\{\mathcal{P}_r\}_{r\in R}$, assume $\mathcal{Q}\preceq \mathcal{P}_r$ for each $r\in R$.  Let $\mathcal{Q}=\{C_k\}_{k\in K}$.  Then, for each
  $k\in K$ and each $r\in R$, there exists $i_r\in I_r$ such that $C_k\subseteq A_{r,i_r}$.  Hence $C_k\subseteq B_{(i_r)}=\bigcap_{r\in R} A_{r,i_r}$.  Since $B_{(i_r)}\in \mathcal{P}$, we have 
  $\mathcal{Q}\preceq \mathcal{P}$.  
  
  The proof of the lemma is complete.  Note that axiom of choice was used in assuming that the set $\mathcal{I}$ is non-empty.

\end{proof}
\textbf{Theorem~\ref{thm:trans-function}. }
 If $\mathcal{P}$ is a congruence, then for any $y\in X$, $\beta \in \overline{\mathcal{T}}(y)$, $\trans$ is a congruence.  
\begin{proof}
Let $\mathcal{P}=\{A_{i}\}_{i\in I}$.  Let  $f(\mathcal{P})=\trans$ and $\{B_i\}_{i\in I}$ be defined as in Definition~\ref{def:trans-function}. 
Since $\mathcal{P}$ is a congruence, for each $t\in \mathcal{T}$, there exists classes $A_i,A_j\in \mathcal{P}$ such that $t\in A_i$ and $t[y\leftarrow \beta]\in A_j$.  
(Note that $i=j$ is possible, but $A_i$ and $A_j$ are uniquely determined by the terms $t$ and $t[y\leftarrow \beta]$).  First we show that $f(\mathcal{P})$ has $t$ in exactly
one class.  
\begin{itemize}
 \item If $t\notin A_i(y)$ (i.e., $y$ does not appear in $t$), then clearly $t\in B_i$ and no other class of $f(\mathcal{P})$ by definition.  In particular, when $t=c$ for any constant $c\in C$, 
 $c\in B_i$ if and only if $c\in A_i$, establishing condition (1) of Definition~\ref{def:cong}.
 \item If $t\in A_i(y)$, then by definition of $B_j$, $t\in B_j$ and no other class in $f(\mathcal{P})$ contains $t$.  
\end{itemize}
We have shown that $f(\mathcal{P})$ is a partition that satisfies condition (1) of Definition~\ref{def:cong}.   Next, we prove that $f(\mathcal{P})$ satisfies the remaining conditions of 
the definition of congruence. \\ 
\emph{Condition (2)}:  we need to prove that for all $t,t',s,s'\in \mathcal{T}$, $t'\confp t$ and $s'\confp s$ if and only if  $t'+s'\confp t+s$.  We have:
\begin{eqnarray*}
&&t'+s'\confp t+s\\
&\iff& (t'+s')[y\leftarrow \beta] \conp (t+s)[y\leftarrow \beta] \text{ (by definition of $\trans$) }\\
&=& t'[y\leftarrow \beta]+s'[y\leftarrow \beta] \conp t[y\leftarrow \beta]+s[y\leftarrow \beta] \text{ (by Definition~\ref{def:substitution})}\\
&\iff& t'[y\leftarrow \beta] \conp t[y\leftarrow \beta] \text{ and } s'[y\leftarrow \beta] \conp s[y\leftarrow \beta] \text{( by definition of }\conp \text{)}\\ 
&\iff& t' \confp t \text{ and } s' \confp s \text{ (by definition of $\trans$) }
\end{eqnarray*}
\emph{Condition (3)}: For any $c\in C$, let $t\in \mathcal{T}$ such that $t\neq c$.   Suppose $c\confp t$, then by Observation~\ref{obs:trans} $t[y\leftarrow \beta]\in X\cup \{c\}$ (by condition (3)
of the definition of congruence), which is possible only if one among the following cases are true: (1) $t=x$ for some $x\in X$, $x\neq y$ such that $x\conp c$ or 
(2) $t=y$ and $\beta=x$ for some $x\in X$, $x\conp c$ or (3) $t=y$ and $\beta=c$.  In any case $t\in X$.  
\end{proof}
\textbf{Lemma~\ref{lem:distrb}. }
 $\transone \wedge \transtwo = \transmeet$. 
\begin{proof}
 If $\mathcal{P}_1 =\top$, then $\transone \wedge \transtwo= \top \wedge \transtwo= \transtwo= f(\top \wedge \mathcal{P}_2) $ and the lemma holds true. 
 Similarly, if $\mathcal{P}_2 =\top$, the lemma holds true. Otherwise, let $\mathcal{P}_1=\{A_i\}_{i\in I}$ and $\mathcal{P}_2=\{B_j\}_{j\in J}$.
 By Definition~\ref{def:trans-function}, we have:
 \begin{eqnarray*}
  \transone&=&\{A'_i : i\in I, A'_i\neq \emptyset\} \text{, where } A'_i=\{t\in \mathcal{T} :  t[y\leftarrow \beta]\in A_i\}\\
  \transtwo&=&\{B'_j : j\in J, B'_j\neq \emptyset\} \text{, where } B'_j=\{t\in \mathcal{T} :  t[y\leftarrow \beta]\in B_j\}
 \end{eqnarray*}
By the definition of meet, we have:  
\begin{eqnarray}\label{eqn:eqdistr}
 \transone \wedge \transtwo &=& \{A'_i\cap B'_j : i\in I, j\in J, A'_i\cap B'_j\neq \emptyset\}\\
 \mathcal{P}_1\wedge \mathcal{P}_2 &=& \{A_i\cap B_j : i\in I, j\in J, A_i\cap B_j\neq \emptyset\}
\end{eqnarray}
Moreover, 
\begin{equation}\label{eqn:eqmeet}
 \transmeet = \{D_{i,j} : i\in I, j\in J, D_{i,j}\neq \emptyset\}
\end{equation}
where, 
\begin{eqnarray*}
D_{i,j}&=&\{t\in \mathcal{T} :  t[y\leftarrow \beta]\in A_i\cap B_j\}\\
&=&\{t\in \mathcal{T} :  t[y\leftarrow \beta]\in A_i \text{ and } t[y\leftarrow \beta]\in B_j\}\\   
&=&\{t\in \mathcal{T} :  t\in A'_i \text{ and } t\in B'_j\}=A'_i\cap B'_j         
\end{eqnarray*}
Thus, Equation~(\ref{eqn:eqmeet}) becomes,
\begin{equation}\label{eqn:eqfinalmeet}
 \transmeet = \{A'_i\cap B'_j  : i\in I, j\in J, A'_i\cap B'_j \neq \emptyset\}
\end{equation}
 The result follows by comparing equations (\ref{eqn:eqdistr}) and (\ref{eqn:eqfinalmeet}). 
\end{proof}
\textbf{Theorem~\ref{thm:continuous}. }
 For any $\emptyset \neq S \subseteq \gtbar$, $f(\bigwedge S)=\bigwedge f(S)$.  
 \begin{figure}[h] 
 \begin{center}
 \includegraphics[scale=0.75]{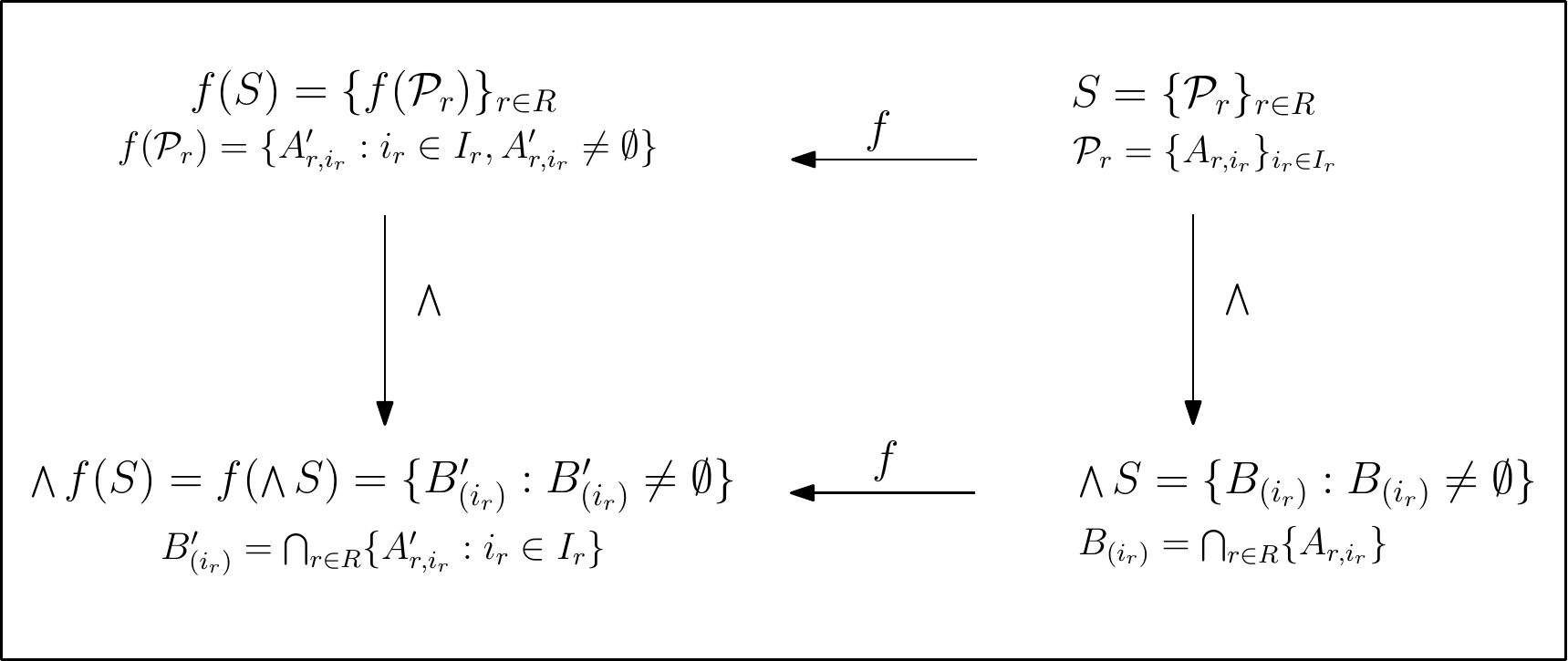}
  \caption{Commutativity diagram for Theorem~\ref{thm:continuous}}
\label{fig:commutativity} 
\end{center}
\end{figure}
 \begin{proof}
 If $S=\{\top\}$, $f(\bigwedge S)=\top=\bigwedge f(S)$.  
 Otherwise, we may assume without loss of generality that $\top \notin S$ (else consider $S\setminus \{\top\}$).  

 Let $S=\{\mathcal{P}_r\}_{r\in R}$ be a non-empty collection of congruences, not containing $\top$.  
 For each $r\in R$, let $\mathcal{P}_r=\{A_{r,i_r}\}_{i_r\in I_r}$ for some index set $I_r\neq \emptyset$.
 By Definition~\ref{def:trans-function} of transfer function,  for each $r\in R$ we have:  
 \begin{eqnarray}\label{eqn:fmeetP}
  f(\mathcal{P}_r)&=&\{A'_{r,i_r} : i_r\in I_r, A'_{r,i_r}\neq \emptyset\} \text{, where } A'_{r,i_r}=\{t\in \mathcal{T} :  t[y\leftarrow \beta]\in A_{r,i_r}\}
 \end{eqnarray} 
 Hence $f(S)=\{ f(\mathcal{P}_r) :  r\in R\}$.   
 
 Let $\mathcal{I}$ be the Cartesian product of the index sets of the congruences; i.e., $\mathcal{I}=\prod_{r \in R}{I_r}$. Each element in $\mathcal{I}$ is a sequence $(i_r)_{r \in R}$.
 For each $(i_r)_{r\in R} \in \mathcal{I}$, define \\$B_{(i_r)}= \bigcap_{r\in R} A_{r,i_r}$ and $B'_{(i_r)}= \bigcap_{r\in R} A'_{r,i_r}$.  As shown in the proof of Lemma~\ref{lem:subsetInfimum},
 we have:  
\begin{equation}\label{eqn:meetS}
 \bigwedge S = \{B_{(i_r)} : (i_r) \in \mathcal{I}, B_{(i_r)} \ne \emptyset \}
\end{equation}
\begin{equation}\label{eqn:meetfS}  
 \bigwedge f(S) = \{B'_{(i_r)} : (i_r) \in \mathcal{I}, B'_{(i_r)} \ne \emptyset \}
\end{equation}
From Equation~\ref{eqn:meetS} and Definition~\ref{def:trans-function} of transfer function we write: 
\begin{eqnarray}\label{eqn:fmeetS}
 f(\bigwedge S)= \{D_{(i_r)} : (i_r) \in \mathcal{I}, D_{(i_r)} \ne \emptyset \} \text{ where, } D_{(i_r)} = \{t\in \mathcal{T}: t[y\leftarrow \beta]\in B_{(i_r)}\}
\end{eqnarray}
Comparing Equation~\ref{eqn:meetfS} and Equation~\ref{eqn:fmeetS}, we see that to complete the proof, it is enough to prove that $D_{(i_r)}= B'_{(i_r)}$ for each $(i_r)\in I$ (see 
Figure~\ref{fig:commutativity}).
\begin{eqnarray*}
D_{(i_r)}&=&\{t\in \mathcal{T}: t[y\leftarrow \beta]\in B_{(i_r)}\}\\
         &=&\{t\in \mathcal{T}: t[y\leftarrow \beta]\in \cap_{r\in R} A_{r,i_r}\}\text{ (by definition of $B_{(i_r)}$) } \\
         &=&\{t\in \mathcal{T}: t[y\leftarrow \beta]\in A_{r,i_r}\text{ for each }r\in R\} \\
         &=&\{t\in \mathcal{T}: t\in A'_{r,i_r}\text{ for each }r\in R\} \text{ (by Equation~\ref{eqn:fmeetP}) }\\ 
         &=&\{t\in \mathcal{T}: t\in \cap_{r\in R} A'_{r,i_r}\}\\
         &=&\{t\in \mathcal{T}: t\in B'_{(i_r)}\} \text{ (by definition of $B'_{(i_r)}$) }\\
         &=& B'_{(i_r)}
\end{eqnarray*}

The proof is now complete.  Note that the axiom of choice was used in the definition of the set $\mathcal{I}$.  
\end{proof}
\textbf{Theorem~\ref{thm:non-det-trans}. }
If $\mathcal{P}$ is a congruence, then for any $y\in X$, $\ntrans$ is a congruence.  
\begin{proof}
We write $f(\mathcal{P})$ instead of $\ntrans$ to avoid cumbersome notation.  
First we show that $\confp$ is an equivalence relation, thereby establishing that $f(\mathcal{P})$ is a partition of $\mathcal{T}$.    
Reflexivity, symmetry and transitivity of $\confp$ is clear from Definition~\ref{def:non-det-trans}.  Moreover, each 
equivalence class in $\ntrans$ is a subset of some equivalence class in the congruence $\mathcal{P}$. Consequently, 
\emph{ Condition (1) } and \emph{ Condition (3) } in the definition of congruence (Definition~\ref{def:cong}) will be satisfied by $f(\mathcal{P})$.  

It suffices to show that $\ntrans$ satisfies Condition (2) of Definition~\ref{def:cong}.  \\
\emph{Condition (2)}: we need to prove that for all $t,t',s,s'\in \mathcal{T}$, $t'\confp t$ and $s'\confp s$ if and only if  $t'+s'\confp t+s$.  We have:
\begin{eqnarray*}
t'+s'&\confp& t+s\\
&\iff& (t'+s')\conp (t+s) \text{ and, for every } \beta \in \mathcal{T}\setminus \mathcal{T}(y),\\
&&(t+s)[y\leftarrow \beta]\conp (t'+s')[y\leftarrow \beta]\text{ (by definition of $\ntrans$) }\\
&\iff& (t'+s')\conp (t+s) \text{ and, for every } \beta \in \mathcal{T}\setminus \mathcal{T}(y),\\
&&(t[y\leftarrow \beta]+s[y\leftarrow \beta])\conp (t'[y\leftarrow \beta]+s'[y\leftarrow \beta])\\
&\iff& (t' \conp t) \text{ and } (s' \conp s) \text{ and, for every } \beta \in \mathcal{T}\setminus \mathcal{T}(y),\\
&&(t[y\leftarrow \beta] \conp t'[y\leftarrow \beta]) \text{ and } s[y\leftarrow \beta]\conp s'[y\leftarrow \beta])\text{ (by Definition~\ref{def:cong}) }\\
&\iff& (t' \confp t) \text{ and } (s' \confp s) \text{ (by definition of $\ntrans$) }\\
\end{eqnarray*}
The proof of the theorem is complete. 
\end{proof}
\textbf{Theorem~\ref{thm:ntrans-char}. }
If $\mathcal{P}$ is a congruence, then for any $y\in X$, \\
\indent $\ntrans=\mathcal{P}\wedge \left( \bigwedge_{\beta \in \overline{\mathcal{T}}(y)} \trans \right)$.
\begin{proof}
Let $\mathcal{Q}=\mathcal{P}\wedge \left( \bigwedge_{\beta \in \overline{\mathcal{T}}(y)} \trans \right)$.
It follows from condition (1) of Definition~\ref{def:non-det-trans} and Theorem~\ref{thm:non-det-trans} that $\ntrans$ is a congruence, and a refinement of $\mathcal{P}$.
That is, $\ntrans \preceq \mathcal{P}$.  
From condition (2) of Definition~\ref{def:non-det-trans}, we have,  
\begin{eqnarray*}
 t\nconfp t' &\implies& \text{ for every }\beta\in \mathcal{T}\setminus \mathcal{T}(y)\text{, }t[y\leftarrow \beta]\conp t'[y\leftarrow \beta] \\
 &\iff& \text{ for every }\beta\in \mathcal{T}\setminus \mathcal{T}(y)\text{, }t \cong_{_{f_{_{y=\beta}}(\mathcal{P})}}  t' \text{ (by Observation~\ref{obs:trans})} \\
\end{eqnarray*}
Hence for each $\beta\in \mathcal{T}\setminus \mathcal{T}(y)$, $\ntrans \preceq \trans$.  It follows from the definition of $\bigwedge$ that  
$\ntrans \preceq  \bigwedge_{\beta \in \overline{\mathcal{T}}(y)} \trans$ . Combining this with the fact that $\ntrans \preceq \mathcal{P}$ and using the definition
of $\wedge$, we get $\ntrans \preceq Q$.

Conversely let $t,t'\in \mathcal{T}$ such that $t\conq t'$.
Then by definition of $\mathcal{Q}$, we have $t\conp t'$ and, for all $\beta \in \mathcal{T}\setminus \mathcal{T}(y)$ $t \cong_{_{f_{_{y=\beta}}(\mathcal{P})}} t'$.  
By the definition of $\ntrans$ (Definition~\ref{def:non-det-trans}), we have $t\nconfp t'$.  Consequently $Q\preceq \ntrans$.  
 
Hence, $\ntrans=\mathcal{Q}$.  Note that the proof involves use of the  axiom of choice in assuming existence of $\bigwedge$.  
\end{proof}
\textbf{Theorem~\ref{thm:ntrans-cts}. }
 For any $\emptyset \neq S \subseteq \gtbar$, ${f}_{_{y=*}}(\bigwedge S)=\bigwedge {f}_{_{y=*}}(S)$, where\\  ${f}_{_{y=*}}(S)=\{f_{_{y=*}}(s): s\in S\}$.
\begin{proof}
To avoid cumbersome notation, we will write $f$ for $f_{_{y=*}}$. 
 If $S=\{\top\}$, $f(\bigwedge S)=\top=\bigwedge f(S)$.  
 Otherwise, we may assume without loss of generality that $\top \notin S$ (else consider $S\setminus \{\top\}$).  
 
 Let $S=\{\mathcal{P}_r\}_{r\in R}$ be a non-empty collection of congruences, not containing $\top$. 
 We denote by $\bigwedge_{r\in R} \mathcal{P}_r$ the congruence $\bigwedge \{\mathcal{P}_r: r\in R\}$ and 
 $\bigwedge_{r\in R}f(\mathcal{P}_r)$ for the congruence $\bigwedge \{ f(\mathcal{P}_r): r\in R\}$.
 
 \begin{eqnarray*}
 f(\bigwedge S)&=& f(\bigwedge_{r\in R} \mathcal{P}_r)\\
               &=& (\bigwedge_{r\in R} \mathcal{P}_r) \wedge ( \bigwedge_{\beta \in \overline{\mathcal{T}}(y)}f(\bigwedge_{r\in R} \mathcal{P}_r) )  \text{ (by Theorem~\ref{thm:ntrans-char}) } \\ 
               &=& (\bigwedge_{r\in R} \mathcal{P}_r) \wedge ( \bigwedge_{\beta \in \overline{\mathcal{T}}(y)}\bigwedge_{r\in R}\transr ) \text{ (by continuity of $f_{y=\beta}$) } \\    
               &=& (\bigwedge_{r\in R} \mathcal{P}_r) \wedge ( \bigwedge_{r\in R} \bigwedge_{\beta \in \overline{\mathcal{T}}(y)}   \transr ) \text{ (by properties of $\bigwedge$) } \\  
               &=& \bigwedge_{r\in R} ( \mathcal{P}_r \wedge  \bigwedge_{\beta \in \overline{\mathcal{T}}(y)}  \transr ) \text{ (by properties of $\bigwedge$) } \\     
               &=& \bigwedge_{r\in R} f(\mathcal{P}_r)  \text{ (by Theorem~\ref{thm:ntrans-char}) } \\
               &=& \bigwedge f(S) 
 \end{eqnarray*}
Note that the axiom of choice was implicitly used in the proof.  
 \end{proof}
\textbf{Lemma~\ref{lem:nondet-constant}. }
 Let $\mathcal{P}\in \gt$.  Let $t,t'\in \mathcal{T}$ and $c_1,c_2\in C$ with $c_1\neq c_2$.  Then 
 $t[y\leftarrow c_1]\conp t'[y\leftarrow c_1]$ and $t[y\leftarrow c_2]\conp t'[y\leftarrow c_2]$ if and only if for every $\beta\in \overline{\mathcal{T}}(y)$,
 $t[y\leftarrow \beta]\conp t'[y\leftarrow \beta]$.  
\begin{proof}
 One direction is trivial.  For the other, assume that $\mathcal{P}\in \gt$ and $c_1,c_2 \in C$, $c_1\neq c_2$.  If $t=t'$ or $t,t'\notin \mathcal{T}(y)$,
 the lemma holds true trivially.  Hence we assume that $t\in \mathcal{T}(y)$ and $t'\neq t$.  If $t'\notin \mathcal{T}(y)$,  Then our assumption leads to 
 $t[y\leftarrow c_1]\conp t'[y\leftarrow c_1]=t'=t'[y\leftarrow c_2]\conp t[y\leftarrow c_2]$.  This is impossible by Lemma~\ref{lem:const1}.  
 Hence we may assume that $t,t'\in \mathcal{T}(y)$. 
 
 If $t=y$, then by Observation~\ref{obs:const2}, $t[y\leftarrow c_1]=c_1\ncong_{_{\mathcal{P}}} t'[y\leftarrow c_1]$ for any $t'\in \mathcal{T}(y)$, $t'\neq y$.  
 Hence we may assume that $t\neq y$.  Similarly, we may assume that $t'\neq y$ as well.   Hence the lemma holds true whenever at least
 one among  $t,t'$ is in $X\cup C$.   We now proceed by induction.  
 
 Suppose $t=t_1+t_2$ and $t'=t'_1+t'_2$. We have:  
\begin{eqnarray*}\label{eqn:const1}
t[y\leftarrow c_1]\conp t'[y\leftarrow c_1] &\iff& t_1[y\leftarrow c_1]+t_2[y\leftarrow c_1] \conp t'_1[y\leftarrow c_1]+t'_2[y\leftarrow c_1] \\
 &\iff& t_1[y\leftarrow c_1] \conp t'_1[y\leftarrow c_1] \text{ and } t_2[y\leftarrow c_1] \conp t'_2[y\leftarrow c_1]\\
 &&\text{ (By condition (2) of Definition~\ref{def:cong})}
\end{eqnarray*}
Similarly,
\begin{eqnarray*}\label{eqn:const2}
t[y\leftarrow c_2]\conp t'[y\leftarrow c_2] &\iff&t_1[y\leftarrow c_2] \conp t'_1[y\leftarrow c_2] \text{ and } t_2[y\leftarrow c_2] \conp t'_2[y\leftarrow c_2]
\end{eqnarray*}
We may assume as induction hypothesis that:  
 \begin{eqnarray*}
 t_1[y\leftarrow c_1] \conp t'_1[y\leftarrow c_1] &\text{ and }& t_1[y\leftarrow c_2] \conp t'_1[y\leftarrow c_2] \\
 &\iff&t_1[y\leftarrow \beta]\conp t'_1[y\leftarrow \beta] \text{, for every }\beta\in \overline{\mathcal{T}}(y)
\end{eqnarray*}
\begin{eqnarray*}
 t_2[y\leftarrow c_1] \conp t'_2[y\leftarrow c_1] &\text{ and }& t_2[y\leftarrow c_2] \conp t'_2[y\leftarrow c_2] \\
 &\iff&t_2[y\leftarrow \beta]\conp t'_2[y\leftarrow \beta] \text{, for every }\beta\in \overline{\mathcal{T}}(y)
\end{eqnarray*}
Since the left sides of the above equivalences hold by assumption, we have:  
\begin{eqnarray*}
 t_1[y\leftarrow \beta]\conp t'_1[y\leftarrow \beta] \text{ and }
 t_2[y\leftarrow \beta]\conp t'_2[y\leftarrow \beta] \text{, for every }\beta\in \overline{\mathcal{T}}(y)
\end{eqnarray*}
Hence, by condition (3) of Definition~\ref{def:cong}, 
$$t[y\leftarrow \beta]\conp t'[y\leftarrow \beta] \text{, for every }\beta\in \overline{\mathcal{T}}(y)$$
\end{proof}
\begin{lemma}\label{lem:const1}
 Let $t\in \mathcal{T}(y)$ and $c_1,c_2\in C$, $c_1\neq c_2$.  Let $\mathcal{P}$ be a congruence.  Then, $t[y\leftarrow c_1]\ncong_{_{\mathcal{P}}} t[y\leftarrow c_2]$.  
 \end{lemma}
\begin{proof}
 If $t=y$, the lemma holds by condition (1) of the definition of congruence (Definition~\ref{def:cong}).  Otherwise, we proceed by induction.   
 Let $t=t_1+t_2$.  Since $t\in \mathcal{T}(y)$, either $t_1\in \mathcal{T}(y)$ or  $t_2\in \mathcal{T}(y)$.  Without loss of generality, we assume that  $t_1\in \mathcal{T}(y)$.  
 We have by induction hypothesis:
 \begin{equation}
  t_1[y\leftarrow c_1] \ncong_{_{\mathcal{P}}} t_1[y\leftarrow c_2]
 \end{equation}
Hence we have:
\begin{eqnarray*}
 t[y\leftarrow c_1]&=& t_1[y\leftarrow c_1]+t_2[y\leftarrow c_1]\\
                   &\ncong_{_{\mathcal{P}}}&  t_1[y\leftarrow c_2]+t_2[y\leftarrow c_2] \text{ (by condition (2) of Definition~\ref{def:cong}) }\\
                   &=& t[y\leftarrow c_2]                   
 \end{eqnarray*}

\end{proof}
\textbf{Theorem~\ref{thm:nondet-const}. }
Let $\mathcal{P}\in \gt$ and let $c_1,c_2\in C$, $c_1\neq c_2$.  Then, for any $y\in X$,\\ 
$\ntrans=\mathcal{P}\wedge f_{_{y\leftarrow c_1}}(\mathcal{P}) \wedge f_{_{y\leftarrow c_2}}(\mathcal{P})$
\begin{proof}
 Let $\mathcal{Q}=\mathcal{P}\wedge f_{_{y\leftarrow c_1}}(\mathcal{P}) \wedge f_{_{y\leftarrow c_2}}(\mathcal{P})$ and let  
 $\mathcal{Q}'=\mathcal{P}\wedge \left( \bigwedge_{\beta \in \overline{\mathcal{T}}(y)} \trans \right)$
 We have:   
\begin{eqnarray*}
t\conq t' &\iff& t\conp t' \text{ and }t \concone t' \text{ and } t \conctwo t' \text{ (by definition of confluence) }\\
           &\iff& t\conp t'\text{ and } t[y\leftarrow c_1]\conp t'[y\leftarrow c_1] \text{ and } t[y\leftarrow c_2]\conp t'[y\leftarrow c_2]\\
           &&\text{ (by Observation~\ref{obs:trans}) } \\
           &\iff & t\conp t' \text{ and } \forall \beta\in \overline{\mathcal{T}}(y)\text{ } t[y\leftarrow \beta]\conp t'[y\leftarrow \beta] \text{ (by Lemma~\ref{lem:nondet-constant}) } \\
           &\iff & t\conp t' \text{ and } \forall \beta\in \overline{\mathcal{T}}(y)\text{ } t \cong_{_{f_{_{y=\beta}}(\mathcal{P})}}  t' \text{ (by Observation~\ref{obs:trans})}\\
           &\iff & t\cong_{\mathcal{Q}'} t' \text{ (by definition of confluence)  } 
\end{eqnarray*}
\end{proof}
\section{Lattice Properties}
Certain standard facts from lattice theory used in the paper are proved here for easy reference.  

In the following,  $(L,\leq, \bot, \top)$ denotes a lattice with greatest element $\top$ and least element $\bot$. (We will refer to the lattice simply as $L$ when there is not scope for confusion).  
For $a,b\in L$, $a\wedge b$ (respectively $a\vee b$) denotes the meet (respectively join) of $a$ and $b$.   For $S\subseteq L$,   $\bigwedge S$ (respectively $\bigvee S$)
denotes the greatest lower bound (respectively least upper bound) of the set $S$ whenever it exists. The lattice $L$ is meet complete if  
$\bigwedge S$ exists for every subset $S$ of $L$.  In particular, $\bigwedge \emptyset = \top$ and $\bigwedge L=\bot$. A lattice $L$ is a complete lattice if $\bigwedge S$ 
and $\bigvee S$ exists for every $S \subseteq L$.
\begin{theorem}\label{Athm:complete}
If $(L,\leq,\bot, \top)$  is meet complete, then it is a complete lattice.   
\end{theorem}
\begin{proof}
It suffices to prove that for each $S\subseteq L$, $\bigvee S$ exists.  Let $S$ be a subset of
 $L$.  Let $T=\{t \in L:  s\leq t \text{ for each } s \in S\}$.  Since $L$ is meet complete there exists $t_0\in L$ such that $t_0=\bigwedge T$.  We claim
 that $t_0$ is the least upper bound of $S$.  By definitions of the set $T$, for each $t\in T$ and $s\in S$, $s\leq t$.  Hence, by the definition of meet, 
 $s\leq t_0$ for each $s\in S$.   Consequently $t_0$ is an upper bound to $S$.   Next, if any $t\in L$ satisfies $s\leq t$ for each $s\in S$,  then by the definition
 of $T$, $t\in T$. Hence $t_0\leq t$ (by the definition of $t_0$).  Thus $t_0$ is the least upper bound of $S$.  
\end{proof}

The following property of the meet operation will be frequently used.  

\begin{lemma}\label{Alem:meet-union}
If $(L,\leq,\bot,\top)$ is a complete lattice.  Let $\{S_i\}_{i\in I}$ be subsets of $L$. Then, $\bigwedge (\bigcup_{i\in I} S_i)=\bigwedge \{ \bigwedge S_i : i\in I\}$.  
\end{lemma}
\begin{proof}
For each $i\in I$, let $\alpha_i=\bigwedge S_i$ and $\alpha=\bigwedge \{\alpha_i : i\in I\}$.  Let $S=\bigcup_{i\in I} S_i$. We have to prove that $\alpha = \bigwedge S$.    
Since $\alpha\leq \alpha_i$ for each $i\in I$, $\alpha \leq s$ for each
$s\in S_i$ for each $i\in I$.  Hence $\alpha$ is a lower bound to $S$.  If $\beta \in L$ satisfies $\beta \leq s$ for each $s\in S$, then $\beta \leq s_i$ for each $s_i\in S_i$
for all $i\in I$.  Hence $\beta \leq \alpha_i$ for each $i\in I$.  Consequently, $\beta \leq \alpha$.  
\end{proof}

We next define the Cartesian product of lattices.

\begin{definition}[Product Lattice]\label{Adef:Prod_lattice}  
Let $(L,\leq,\bot, \top)$ be a lattice and $n$ a positive integer.  The product lattice, $(L^{n},\leq_n,\bot_n,\top_n)$ is defined as follows: 
for $\overline{a}=(a_1, a_2,\ldots, a_n)$, 
$\overline{b}=(b_1, b_2,\ldots, b_n) \in L^{n}$, $\overline{a}\leq_n \overline{b}$ if $a_i\leq b_i$ for each $1\leq i\leq n$, 
$\bot_n=(\bot, \bot, \ldots,\bot)$ and $\top_n=(\top,\top, \ldots, \top)$.
\end{definition}

It is easy to see that $L^{n}$ is a lattice with the meet of $\overline{a}$ and $\overline{b}$ given by $(a_1\wedge b_1, a_2\wedge b_2,\ldots, a_n\wedge b_n)$.  
The meet of $\overline{a}$ and $\overline{b}$ will be denoted by  $\overline{a}\wedge_n \overline{b}$.  Similarly,
the join of $\overline{a}$ and $\overline{b}$ is $\overline{a}\vee_n \overline{b}=(a_1\vee b_1, a_2\vee b_2,\ldots, a_n\vee b_n)$.  For $\tilde{S}\subseteq L^{n}$, the notation
$\bigwedge_n \tilde{S}$ (respectively $\bigvee_n \tilde{S}$) denotes the greatest lower bound (respectively least upper bound) of $\tilde{S}$ in $L^{n}$ whenever it exists.  

\begin{theorem}\label{Athm:prod_complete}
If $(L,\leq,\bot,\top)$ is a complete lattice, then  $(L^{n},\leq_n,\bot_n,\top_n)$ is a complete lattice. 
\end{theorem}
\begin{proof}
In view of Theorem~\ref{Athm:complete}, it suffices to prove that for each $\tilde{S}\subseteq L^{n}$, $\bigwedge_n \tilde{S}$ exists in $L^{n}$.  If $\tilde{S}=\emptyset$, the proof is
trivial.  
Let $\tilde{S}=S_1\times S_2\times \cdots \times S_n$, where $S_i\subseteq L$ for $1\leq i\leq n$.  We will show that $\bigwedge_n \tilde{S}=(\bigwedge S_1,\bigwedge S_2, \ldots, \bigwedge S_n)$.
Let $\alpha_i=\bigwedge S_i$ for $1\leq i\leq n$. Let $\overline{\alpha}=(\alpha_1,\alpha_2,\ldots, \alpha_n)$.  Since $\alpha_i\leq s_i$ for each $s_i\in S_i$, by the definition of $\leq_n$, 
we have $\overline{\alpha}\leq_n \overline{s}$ for each $\overline{s}\in \tilde{S}$.   Suppose $\overline{\beta}=(\beta_1,\beta_2,\ldots,\beta_n)\in L^{n}$ satisfies 
$\overline{\beta}\leq_n \overline{s}$ for each $\overline{s}\in \tilde{S}$.  Then, by the definition of $\leq_n$, we have $\beta_i\leq s_i$ for each $s_i\in S_i$, $1\leq i\leq n$.  It follows from the
definition of $\alpha_i$ that $\beta_i\leq \alpha_i$ for all $1\leq i\leq n$.  Consequently, $\beta \leq_n \alpha$.  The proof is complete.  
\end{proof}

The proof of Theorem~\ref{Athm:prod_complete} yields the following corollary.  

\begin{corollary}\label{Acor:prod_complete} 
Let $\tilde{S}\subseteq L^{n}$ be non-empty with $\tilde{S}=S_1\times S_2\times \cdots \times S_n$, where $S_i\subseteq L$ for $1\leq i\leq n$.  
Then $\bigwedge_n \tilde{S}=(\bigwedge S_1,\bigwedge S_2, \ldots, \bigwedge S_n)$.
\end{corollary}

Next we define continuous maps between lattices and show that continuous maps are distributive and monotone.   
Let $(L,\leq,\bot,\top)$ and $(L'\leq',\bot',\top')$ be complete lattices.  For arbitrary subsets $S\subseteq L$ and $S'\subseteq L'$, we 
use the notation $\bigwedge S$ and $\bigwedge' S'$ to denote the greatest lower bounds of $S$ and $S'$ in the respective lattices 
$L$ and $L'$.  

\begin{definition}\label{Adef:cts}
Let $(L,\leq,\bot,\top)$ and $(L'\leq',\bot',\top')$ be complete lattices. A function $f: L\rightarrow L'$ is continuous 
if for each $\emptyset \neq S\subseteq L$, $f(\bigwedge S)=\bigwedge' f(S)$, where $f(S)=\{f(s): s\in S\}$.
\end{definition}

\begin{obs}\label{Aobs:monotone}
Let $(L,\leq,\bot,\top)$ and $(L'\leq',\bot',\top')$ be complete lattices. A function $f: L\rightarrow L'$ be continuous.  Let $s,s'\in L$ be chosen arbitrarily. 
Then $f$ satisfies:  
\begin{enumerate}
 \item Distributivity:  $f(s\wedge s')=f(s)\wedge f(s')$.
 \item Monotonicity:  if $s\leq s'$ then $f(s)\leq f(s')$.  
\end{enumerate}
The first property is immediate from the definition of continuity.  For the second, assume that $s\leq s'$. Then $s=s\wedge s'$ and we have  $f(s)=f(s\wedge s')=f(s)\wedge f(s')\leq f(s')$.  
\end{obs}

We next show that two particular families of maps from $L^{n}$ to $L$ called projection maps and confluence maps are continuous.

\begin{definition}[Projection Maps]\label{Adef:proj_map}
 Let $(L,\leq,\bot,\top)$ be a lattice. For each $i\in \{1,2,\ldots,n\}$,  the projection map to the $i^{th}$ co-ordinate, 
 $\pi_i:L^{n}\mapsto L$ is defined by $\pi_i(s_1,s_2,\ldots,s_n)=s_i$ for any $(s_1,s_2,\ldots,s_n)\in L^{n}$.  
\end{definition}

\begin{definition}[Confluence Maps]\label{Adef:conf_map}
Let $(L,\leq,\bot,\top)$ be a lattice. For each $i,j\in \{1,2,\ldots,n\}$, the confluence map $\pi_{i,j}:L^{n}\mapsto L$ is defined
by $\pi_{i,j}(s_1,s_2,\ldots, s_n)=s_i\wedge s_j$ for any $(s_1,s_2,\ldots, s_n)\in L^{n}$.  
  
\end{definition}

Note that $\pi_{i,j}(\overline{s})=\pi_i(\overline{s})\wedge \pi_j(\overline{s})$ and $\pi_{i,i}(\overline{s})=\pi_i(\overline{s})$ for 
all $\overline{s}\in L^{n}$ and $i,j\in \{1,2,\ldots, n\}$.  Thus projection maps are special instances of confluence maps.

\begin{theorem}\label{Athm:conf_cts}
Projection maps and confluence maps over a complete lattice $(L,\leq,\bot,\top)$ are continuous.
\end{theorem}
\begin{proof}
 Let $i,j\in\{1,2,\ldots,n\}$.  Let $\emptyset \neq \tilde{S}\subseteq L^{n}$.  We need to prove that $\pi_{i,j}(\bigwedge_n \tilde{S})=\bigwedge \pi_{i,j}(\tilde{S})$.  
 Let $\tilde{S}=S_1\times S_2\times \cdots \times S_n$, $S_i\neq \emptyset$. Let $\alpha=\bigwedge S_i$ and $\beta =\bigwedge S_j$.  
 By Corollary~\ref{Acor:prod_complete} and the definition of $\pi_{i,j}$ we have:
 \begin{eqnarray}
  \pi_{i,j}(\bigwedge_n \tilde{S})&=&\pi_{i,j}(\bigwedge S_1,\bigwedge S_2, \ldots, \bigwedge S_n)=\alpha \wedge \beta \\
  \bigwedge \pi_{i,j}(\tilde{S})&=&\bigwedge \pi_{i,j}(S_1\times S_2\times \cdots \times S_n)=\bigwedge \{s_i\wedge s_j : s_i\in S_i, s_j\in S_j\}
 \end{eqnarray}
 The equality between the left sides of the two equations above follows from Lemma~\ref{Alem:meet-lemma}.
\end{proof}

\begin{lemma}\label{Alem:meet-lemma}
Let $(L,\leq, \bot, \top)$ be a complete lattice.  Let $S_1$ and $S_2$ be non empty subsets of $L$ with $\alpha =\bigwedge S_1$ and $\beta =\bigwedge S_2$.  
Let $S'=\{s_1\wedge s_2: s_1\in S_1, s_2\in S_2\}$.  Then $\bigwedge S'=\alpha \wedge \beta$.
\end{lemma}
\begin{proof}
Since $\alpha$ and $\beta$ are lower bounds to $S_1$ and $S_2$,  we have $\alpha \wedge \beta \leq \alpha \leq s_1$ and 
$\alpha\wedge \beta \leq \beta \leq s_2$, for each $s_1\in S_1$ and $s_2\in S_2$.   Consequently, from the definition of meet,  
we have $\alpha\wedge \beta \leq s_1\wedge s_2$ for any $s_1\in S_1$, $s_2\in S_2$
and thus $\alpha\wedge \beta$ is a lower bound to the set $S'$.  Now, if $\gamma \leq s_1\wedge s_2$ for all $s_1\in S_1$ and $s_2\in S_2$, then clearly 
$\gamma \leq s_1$ for all $s_1\in S_1$ and $\gamma \leq s_2$ for all $s_2\in S_2$.  Since $\alpha=\bigwedge S_1$ and $\beta =\bigwedge S_2$, we have $\gamma \leq \alpha$
and $\gamma \leq \beta$.  Consequently, by the definition of meet, $\gamma \leq \alpha\wedge \beta$.  This shows that $\alpha\wedge \beta$ is the greatest lower bound
of the set $S'$.  
\end{proof}

Next we show that continuous maps are closed under composition.  

\begin{theorem}[Composition]\label{Athm:compos}  
Let $(L,\leq,\bot,\top)$, $(L',\leq',\bot',\top')$ and $(L'',\leq'',\bot'',\top'')$ be complete lattices.  Let $f:L\mapsto L'$ and $g:L'\mapsto L''$ be continuous maps.
Then the composition map, $g\circ f: L\mapsto L''$, defined by $(g\circ f)(s)=g(f(s))$ for each $s\in L$, is continuous.
\end{theorem}
\begin{proof}
Let $\emptyset \neq S\subseteq L$.  We need to prove that $(g\circ f)(\bigwedge S)=\bigwedge (g\circ f)(S)$.  
Let $f(S)=\{f(s): s\in S\}$ and $(g\circ f)(S)=g(f(S))=\{ (g\circ f)(s) : s\in S\}$.    
\begin{eqnarray*}
 (g\circ f)(\bigwedge S) &=& g(f(\bigwedge S))\\
                         &=& g(\bigwedge f(S)) \text{ (By the continuity of $f$) } \\
                         &=& \bigwedge g(f(S)) \text{ (By the continuity of $g$) } \\
                         &=& \bigwedge (g\circ f)(S).
\end{eqnarray*}
\end{proof}

The next theorem shows that $f$ is continuous if and only if each of its component maps are continuous.  We first note the following:

\begin{obs}\label{Aobs:component}
Let $(L,\leq,\bot,\top)$ be a complete lattice. For each $i\in \{1,2,\ldots, n\}$,
let $f_i:L^{n}\mapsto L$ be arbitrary functions.  Then, the function $f:L^{n}\mapsto L^{n}$ defined by 
$f(\overline{s})=(f_1(\overline{s}),f_2(\overline{s}),\ldots, f_n(\overline{s}))$ satisfies $\pi_i\circ f=f_i$ for each $1\leq i\leq n$.
Conversely, for any function  $f:L^{n}\mapsto L^{n}$, the $i^{th}$ component map $f_i:L^{n}\mapsto L$ defined by  
$f_i=\pi_i\circ f$ for each $i\in\{1,2,\ldots, n\}$ satisfies $f_i(\overline{s})=\pi_i(f(\overline{s}))$ for all $\overline{s}\in L^{n}$.  
\end{obs}

\begin{theorem}\label{Athm:prod_cts}
Let $(L,\leq,\bot,\top)$ be a complete lattice.  The map $f:L^{n}\mapsto L^{n}$ is continuous if and only if $f_i=\pi_i\circ f$ is continuous for each 
$1\leq i\leq n$.  
\end{theorem}
\begin{proof}
Suppose $f$ is continuous.  Since $\pi_i$ is continuous (by Theorem~\ref{Athm:conf_cts}) for each $1\leq i\leq n$, by Theorem~\ref{Athm:compos}, $\pi_i\circ f$ is continuous.   
This establishes one direction of the theorem.  

Conversely, suppose $f_i$ is continuous for each $i\in \{1,2,\ldots,n\}$.  Let $\emptyset \neq \tilde{S}\subseteq L^{n}$ be chosen arbitrarily.  We need to prove that 
$f(\bigwedge_n \tilde{S})=\bigwedge_n f(\tilde{S})$.  Let $\tilde{S}=S_1\times S_2\times \cdots \times S_n$, where $S_i \subseteq L$ for each $1\leq i\leq n$.  Let $\alpha_i=\bigwedge S_i$
for each $1\leq i\leq n$.  Let $\overline{\alpha} =(\alpha_1,\alpha_2,\ldots, \alpha_n)$.  By Corollary~\ref{Acor:prod_complete} we have:
\begin{eqnarray}\label{Aeqn:alpha}
 \bigwedge_n \tilde{S}=(\bigwedge S_1,\bigwedge S_2,\ldots,\bigwedge S_n)=(\alpha_1,\alpha_2,\ldots, \alpha_n)=\overline{\alpha} 
\end{eqnarray}
Hence, 
\begin{eqnarray}
 f(\bigwedge_n \tilde{S})&=&f(\overline{\alpha})
\end{eqnarray}
Since for each $\overline{s}\in \tilde{S}$, $f(\overline{s})=(f_1(\overline{s}),f_2(\overline{s}),\ldots, f_n(\overline{s}))$, the above equation becomes:  
\begin{eqnarray}\label{Aeqn:LHS-II}
 f(\bigwedge_n \tilde{S})&=&(f_1(\overline{\alpha}),f_2(\overline{\alpha}),\ldots, f_n(\overline{\alpha}))
\end{eqnarray}
Now,
\begin{eqnarray*}
 \bigwedge_n f(\tilde{S})&=& \bigwedge_n \{ (f_1(\overline{s}), f_2(\overline{s}),\ldots, f_n(\overline{s})) : \overline{s}\in \tilde{S} \}\\
                         &=& \bigwedge_n f_1(\tilde{S}) \times f_2(\tilde{S})\cdots \times f_n(\tilde{S}) \text{ (by definition of Cartesian product) }\\
                         &=& (\bigwedge f_1(\tilde{S}),\bigwedge f_2(\tilde{S}),\ldots, \bigwedge f_n(\tilde{S})) \text{ (by Corollary~\ref{Acor:prod_complete}) }\\
                         &=& (f_1(\bigwedge _n \tilde{S}),f_2(\bigwedge _n\tilde{S}),\ldots, f_n(\bigwedge _n \tilde{S})) \text{ (by continuity of $f_1,f_2,\ldots, f_n$) }\\
                         &=& (f_1(\overline{\alpha}),f_2(\overline{\alpha}),\ldots, f_n(\overline{\alpha})) \text{ (by Equation~\ref{Aeqn:alpha}) }\\
\end{eqnarray*}
The theorem follows by comparing the above equation with Equation~\ref{Aeqn:LHS-II}.
\end{proof}

Next we turn to the computation of the maximum fix-point of a continuous map on a complete lattice. The well known Knaster-Tarski Theorem \cite{tarski1955} asserts the existence of a maximum fix-point for every 
monotone function defined over a complete lattice.  When the function is continuous, the maximum fix-point can be defined as the greatest lower bound of a specific iteratively defined 
subset of the lattice, as described below.  In the following we use the notation $f^{2}=f\circ f$, $f^{3}=f\circ f^2$, etc.  

\begin{definition}[Maximum Fix-Point]
 Let $(L,\leq,\bot,\top)$ be a complete lattice.  Let $f:L\mapsto L$ be any function.  $s\in L$ is a maximum fix point of $f$ if:
 \begin{itemize}
  \item $s$ is a fix point of $f$.  That is, $f(s)=s$.
  \item for any $s'\in L$, if $f(s')=s'$, then $s'\leq s$.  
 \end{itemize}

\end{definition}

\begin{theorem}\label{Athm:MFP}
Let $(L,\leq,\bot,\top)$ be a complete lattice. Let $f:L\mapsto L$ be continuous.  Let $S=\{\top, f(\top), f^{2}(\top), f^{3}(\top),\ldots\}$.  Then $\bigwedge S$ is the maximum fix point of $f$.  
\end{theorem}
\begin{proof} 
Let $s_0=\bigwedge S$. We need to prove that $s_0$ is the maximum fix point of $f$.  If $f(\top)=\top$, then $S=\{\top\}$ and the theorem holds.  Otherwise,
as $f(\top)\leq \top$, it follows from the monotonicity of $f$ (Observation~\ref{Aobs:monotone}) that, $S$ is a descending chain with 
$f^{i+1}(\top)\leq f^{i}(\top)$ for each $i\geq 1$; and we have $f(S)=S\setminus \{\top \}\neq \emptyset$. By continuity of $f$, we get:
$$f(s_0)=f(\bigwedge S)=\bigwedge f(S)=\bigwedge (S\setminus \{\top\})=\bigwedge S=s_0$$  
This shows $s_0$ is a fix-point.  Suppose $s'\in L$ satisfies $f(s')=s'$.  Since $s'\leq \top$, $s'=f(s')\leq f(\top)$. Extending the argument, we see that $s'\leq f^{i}(\top)$ for each $i\geq 1$.  
Hence by the definition of $\bigwedge$, $s'\leq \bigwedge S=s_0$.  This proves that $s_0$ is the maximum fix-point of $f$.  
\end{proof}

The following is a consequence of the above proof.  

\begin{corollary}\label{Acor:MFP}
Let $(L,\leq,\bot,\top)$ be a complete lattice. Let $f:L\mapsto L$ be continuous. Then the set $S'=\{\top, f(\top), f^{2}(\top), f^{3}(\top),\ldots\}$ is a decreasing sequence of lattice values 
with $f^{i+1}(\top)\leq f^{i}(\top)$ for each $i\geq 0$.  Moreover, $f(\bigwedge S')$ is the maximum fix point of $f$.  
\end{corollary}

\section{An Algorithm for Computing Program Expression Equivalence}\label{sec:appendixc}
A program expression in a program $P$ is a term over $(X \cup C)$ that actually appear in the program $P$. Two program expressions $e$ and $e'$ are Herbrand equivalent at a program point if
they belong to the same partition class of the Herbrand Congruence at that program point. The \emph{Herbrand equivalence classes of program expressions} at a program point are the partition classes obtained by
restricting the Herbrand Congruence of that point to only the set of program expressions. In this section we describe an algorithm that calculates the Herbrand equivalence classes of program expressions
at each program point in any input program, represented as a data flow framework. We restrict our attention to only intra procedural analysis.
\subsection{Description of Data Structures}
An $ID$ (which stands for a value identifier) is a composite data type with four fields, namely $ftype$, $valueNum$, $idOperand1$ and $idOperand2$. If $ftype$ is $0$, we call it an atomic ID and 
it will have an associated value number called $valueNum$, which is a positive integer and the other two fields $idOperand1$, $idOperand2$ are set to $NIL$. For a non-atomic ID, its $ftype$ will be an operator and 
its $valueNum$ field will be set to $-1$. The field $idOperand1$ will point to the $ID$ of the first operand and $idOperand2$ will point to the $ID$ of the second operand.
An $IdArray$ type represents an array of $ID$s that holds one index for each element $t$ in $(X \cup C) \cup [(X \cup C) \times (X \cup C) \times Op]$ in an order arbitrarily fixed in the beginning. 
Each array element indicates the value identifier of the corresponding element $t$.   

Associated with each program point $p$, there is an $IdArray$ and this together forms an array $Partitions$ which has one index corresponding to each program point. 
This way, for each program point $i$, $Partitions[i]$ will be an array of value identifiers, one index corresponding to each program expression and 
two program expressions $t$ and $t'$ are considered equivalent at program point $i$ if and only if 
$Partitions[i][t]=Partitions[i][t']$. The notation $numClasses$ represents the cardinality of the set $(X \cup C) \cup ((X \cup C) \times (X \cup C) \times Op)$ 
and $numProgPoints$ represents the number of program points in the input program.
\begin{algorithm}
    \caption{Data Structures}
    \label{alg:dataStruct}
    \begin{algorithmic}[1] 
            \State typedef struct \{                 
                 \begin{quote}
                 \vspace{-0.5cm}
                    \State  int ftype; \Comment{$0$ if atomic and the operator number otherwise.}
                    \State  int valueNum; \Comment{value number of the element if ftype is $0$ and $-1$ otherwise.}
                    \State  struct ID* idOperand1; \Comment{$NIL$ if ftype is $0$}
                    \State  struct ID* idOperand2; \Comment{$NIL$ if ftype is $0$}                  
                 \end{quote}
               \State  \} ID;
            \State typedef $IdArray$ $ID[numClasses]$; 
            \State $IdArray$ $Partitions[numProgPoints]$; 
    \end{algorithmic}
\end{algorithm}
\subsection{Description of the Algorithm}
\begin{algorithm}
    \caption{Main Program}
    \label{alg:main}
    \begin{algorithmic}[1] 
        \Procedure{Main}{} 
         \State $StartCounter()$;
         \For{ each term $t \in  X \cup C$ } 
               \State $Partitions[1][t] \gets CreateAtomicId()$;
         \EndFor    
         \For{ each term $t$ of the form $t=t_1 + t_2 $, where $t_1, t_2 \in X \cup C$ and $+$ in $Op$} 
               \State $Partitions[1][t] \gets AssignCompoundId(\&Partitions[1][t_1],\&Partitions[1][t_2],+)$; 
         \EndFor 
         \For{ each program point $k$ from $2$ to $numProgPoints$} 
                   \State $Partitions[k]=\top$;
         \EndFor          
         \State $ConvergeFlag \gets 0$;
         \While{ $ConvergeFlag$ is $0$}
             \State $ConvergeFlag \gets 1$;
             \For{ each program point $k$ from $2$ to $numProgPoints$} 
                   \State $OldPartition \gets Partitions[k]$;
               \If{ $k$ is a function point with $\pred(k)=\{j\}$ and assignment $y \leftarrow \beta$ }                  
                    \State $Partitions[k] \gets AssignStmt(Partitions[j], y, \beta)$;                    
               \ElsIf{ $k$ is a function point with $\pred(k)=\{j\}$ and assignment $y \leftarrow *$ }                    
                    \State $Partitions[k] \gets NonDetAssign(Partitions[j], y)$;    
               \ElsIf{ $k$ is a confluence point with $\pred(k)=\{i, j\}$}     
                    \State $Partitions[k] \gets Confluence(Partitions[i],Partitions[j])$;                    
               \EndIf
               \If{ $IsSame(OldPartition, Partitions[k])$ is $0$}
                           $ConvergeFlag \gets 0$; 
               \EndIf
             \EndFor  
         \EndWhile
            
        \EndProcedure
    \end{algorithmic}
\end{algorithm}

\begin{algorithm}
    \caption{Handling Assignment Statement $y \leftarrow \beta$}
    \label{alg:assgn}
    \begin{algorithmic}[1] 
        \Procedure{AssignStmt}{$IdArray$ $\mathcal{P}$, variable $y$, term $\beta$} \Comment{ Assume that $y$ does not appear in $\beta$.}
            \State $IdArray$ $\mathcal{Q}$;
            \State  Initialize $\mathcal{Q} = \mathcal{P}$   
            \State  $\mathcal{Q}[y] \gets \mathcal{P}[\beta]$
            \For{ each term $t$ of the form $t=y+t'$, where $t' \in X \cup C$} 
                    \State $\mathcal{Q}[t] \gets  AssignCompoundId(\&\mathcal{Q}[y], \&\mathcal{Q}[t'], +)$; 
            \EndFor      
            \For{ each term $t$ of the form $t=t'+y$, where $t' \in X \cup C$} 
               \State $\mathcal{Q}[t] \gets AssignCompoundId(\&\mathcal{Q}[t'], \&\mathcal{Q}[y], +)$; 
            \EndFor   
             \State \textbf{return} $\mathcal{Q}$; 
        \EndProcedure
    \end{algorithmic}
\end{algorithm}
\begin{algorithm}
    \caption{Handling Non-Deterministic Assignment Statement $y \leftarrow *$}
    \label{alg:non-det-assgn}
    \begin{algorithmic}[1] 
        \Procedure{NonDetAssign}{$IdArray$ $\mathcal{P}$, variable $y$} \Comment{}
            \State $IdArray$ $\mathcal{Q}_1,  \mathcal{Q}_2$,  $\mathcal{Q}$;
            \State  Initialize $\mathcal{Q}_1=\mathcal{P}$; $\mathcal{Q}_2=\mathcal{P}$;    
            \State  Let $c_1, c_2 \in C$. 
            \State  $\mathcal{Q}_1[y]=\mathcal{P}[c_1]$;  $\mathcal{Q}_2[y]=\mathcal{P}[c_2]$; 
            \For{ each term $t$ of the form $t=y+t'$, where $t' \in X \cup C$} 
                    \State $\mathcal{Q}_1[t] =$ AssignCompoundId$(\&\mathcal{Q}_1[y], \&\mathcal{Q}_1[t'], +)$; 
                    \State $\mathcal{Q}_2[t] =$ AssignCompoundId$(\&\mathcal{Q}_2[y], \&\mathcal{Q}_2[t'], +)$; 
            \EndFor      
            \For{ each term $t$ of the form $t=t'+y$, where $t' \in X \cup C$} 
               \State $\mathcal{Q}_1[t] =$ AssignCompoundId$(\&\mathcal{Q}_1[t'], \&\mathcal{Q}_1[y], +)$; 
               \State $\mathcal{Q}_2[t] =$ AssignCompoundId$(\&\mathcal{Q}_2[t'], \&\mathcal{Q}_2[y], +)$; 
            \EndFor  
             \State $\mathcal{Q}=Confluence(\mathcal{Q}_1, \mathcal{Q}_2)$;
             \State \textbf{return} $Confluence(\mathcal{Q}, \mathcal{P})$; 
        \EndProcedure
    \end{algorithmic}
\end{algorithm}
\begin{algorithm}
    \caption{Handling Confluence}
    \label{alg:conf}
    \begin{algorithmic}[1] 
        \Procedure{Confluence}{$IdArray$ $\mathcal{P}_1$, $IdArray$ $\mathcal{P}_2$} \Comment{}
            \State $IdArray$ $\mathcal{Q}$; $int$ $AccessFlag[numClasses]$;
             \For{ each term $t \in X \cup C$}
                  \State $AccessFlag[t]=0$;
             \EndFor     
             \For{ each term $t \in X \cup C$} 
                  \If{ $AccessFlag[t]$ is equal to $0$}
                   \State $AccessFlag[t]=1$;
                   \If{ $\mathcal{P}_1[t]$ is equal to $\mathcal{P}_2[t]$}
                         \State $\mathcal{Q}[t]=\mathcal{P}_1[t]$;
                   \Else                    
                         \State $S_1 = getClass(t, \mathcal{P}_1)$; 
                         \State $S_2 = getClass(t, \mathcal{P}_2)$;
                         \State $newId =CreateAtomicId()$;
                         \For{ each term $t'$ in $S_1 \cap S_2 \cap (X \cup C)$}
                               \State $AccessFlag[t']=1$;
                               \State $\mathcal{Q}[t']=newId$;
                         \EndFor      
                   \EndIf
               \EndIf    
            \EndFor
            \For{ each term $t$ of the form $t=t_1+t_2$, where $t_1, t_2 \in X \cup C$} 
                  \State $\mathcal{Q}[t] =$ AssignCompoundId$(\&\mathcal{Q}[t_1], \&\mathcal{Q}[t_2], +)$; 
            \EndFor      
           \State \textbf{return} $\mathcal{Q}$; 
        \EndProcedure
    \end{algorithmic}
\end{algorithm}

\begin{algorithm}
    \caption{Compare Congruences}
    \label{alg:compare}
    \begin{algorithmic}[1] 
        \Procedure{IsSame}{$IdArray$ $\mathcal{P}_1$, $IdArray$ $\mathcal{P}_2$} \Comment{See if two congruences are the same, considering class value-id to be abstract}
            \State $Set$ $S_1, S_2$;
            \For{ each term $t$ in $(X \cup C) \cup ((X \cup C) \times (X \cup C) \times Op)$}
                  \State $S_1 = GetClass(t, \mathcal{P}_1)$;     
                  \State $S_2 = GetClass(t, \mathcal{P}_2)$;
                  \If{ $S_1$ is not equal to $S_2$}
                       \State \textbf{return} $0$;
                  \EndIf     
             \EndFor
             \State \textbf{return} $1$;
        \EndProcedure
    \end{algorithmic}
\end{algorithm}

\begin{algorithm}
    \caption{Get the class of $t$ in the partition $\mathcal{P}$}
    \label{alg:getClass}
    \begin{algorithmic}[1] 
        \Procedure{GetClass}{term $t$, $IdArray$ $\mathcal{P}$} \Comment{Get the partition class of terms congruent to $t$ in the partition $\mathcal{P}$}
            \State $Set$ $S = \emptyset$; 
            \For{ each term $t'$ in $(X \cup C) \cup ((X \cup C) \times (X \cup C) \times Op)$}
                  \If{$\mathcal{P}[t']$ is equal to $\mathcal{P}[t]$}
                       \State $S = S \cup \{t'\}$;
                  \EndIf     
             \EndFor
             \State \textbf{return} $S$;
        \EndProcedure
    \end{algorithmic}
\end{algorithm}

\begin{algorithm}
    \caption{Create a New Atomic Value Id}
    \label{alg:assignAtomicId}
    \begin{algorithmic}[1] 
        \Procedure{CreateAtomicId}{} \Comment{Creates a new atomic value class and assigns it to term $t$}
            \State struct ID $idClass$;
            \State $idClass.valueNum \gets IncCounter()$;
            \State  $idClass.ftype \gets 0$;
            \State $idClass.idOperand1 \gets NIL$; 
            \State $idClass.idOperand2 \gets NIL$;
            \State \textbf{return} $idClass$;
        \EndProcedure
    \end{algorithmic}
\end{algorithm}

\begin{algorithm}
    \caption{Assign a Compound Value Id to a Term}
    \label{alg:assignCompoundId}
    \begin{algorithmic}[1] 
        \Procedure{AssignCompoundId}{struct ID* $id1$, struct ID* $id2$, int $optype$} \Comment{Create a compound value class}
            \State struct ID $idClass$;
            \State $idClass.valueNum \gets-1$    
             \State  $idClass.ftype \gets optype$;
            \State $idClass.idOperand1 \gets id1$; 
            \State $idClass.idOperand2 \gets id2$;
             \State \textbf{return} $idClass$;
        \EndProcedure
    \end{algorithmic}
\end{algorithm}

\begin{algorithm}
    \caption{Initialize Global Counter}
    \label{alg:startCounter}
    \begin{algorithmic}[1] 
        \Procedure{StartCounter}{} \Comment{A global counter is initialized to $0$}
            \State $counter \gets 0 $
        \EndProcedure
    \end{algorithmic}
\end{algorithm}

\begin{algorithm}
    \caption{Increment Global Counter}
    \label{alg:incCounter}
    \begin{algorithmic}[1] 
        \Procedure{IncCounter}{} \Comment{A global counter is incremented on each function invocation}
            \State $counter \gets counter +1 $
            \State \textbf{return} $counter$
        \EndProcedure
    \end{algorithmic}
\end{algorithm}
\end{document}